\documentclass{siamonline171218}

\usepackage{tikz,pgfplots}
\usepackage{lipsum}
\usepackage{amsfonts,amssymb,verbatim}
\usepackage{graphicx}
\usepackage{epstopdf}
\usepackage{algorithm,algpseudocode}
\usepackage{mathtools}

\usepackage{geometry}
\usepackage[shadow,prependcaption,textsize=small]{todonotes}
\usepackage{dsfont}
\ifpdf
  \DeclareGraphicsExtensions{.eps,.pdf,.png,.jpg}
\else
  \DeclareGraphicsExtensions{.eps}
\fi

\graphicspath{{./images/}}

\pgfplotsset{compat=1.16}

\definecolor{pcn}{rgb}{0.19, 0.84, 0.78} 
\definecolor{fes}{rgb}{0.52, 0.8, 0.0} 
\definecolor{hybrid-proj}{rgb}{0.03, 0.15, 0.4}   
\definecolor{hybrid}{rgb}{0.25, 0.41, 0.88}  
\usepackage{cite}

\usepackage[shortlabels]{enumitem}
\setlist[enumerate]{leftmargin=.5in}
\setlist[itemize]{leftmargin=.5in}

\newcommand{\hybrid}{SAFES }
\newcommand{\hybridname}{Subspace Adapting Functional Ensemble Sampler}
\newcommand{\hybridp}{SAFES-P }
\newcommand{\hybridpname}{Subspace Adapting Functional Ensemble Sampler--Projected}

\newsiamremark{remark}{Remark}
\newsiamremark{hypothesis}{Hypothesis}
\crefname{hypothesis}{Hypothesis}{Hypotheses}
\newsiamthm{claim}{Claim}
\newsiamthm{example}{Example} %
\newsiamthm{assumptions}{Assumptions} %

\usepackage{amsopn}

\newcommand{\be}{\begin{equation}}
\newcommand{\ee}{\end{equation}}
\newcommand{\beqas}{\begin{eqnarray*}}
\newcommand{\eeqas}{\end{eqnarray*}}

\newcommand{\dee}{\mathrm{d}}

\newcommand{\iid}{\overset{\mathrm{iid}}{\sim}}

\newcommand{\cG}{\mathcal{G}}

\def \N {{\mathbb N}}
\def \R {{\mathbb R}}

\def \P {{\mathbb P}}

\def \tr {{\mathrm{tr}}}

\newcommand{\G}{\mathcal{G}}

\numberwithin{theorem}{section}

\newcommand{\TheTitle}{A gradient-free subspace-adjusting ensemble sampler for\\ infinite-dimensional Bayesian inverse problems}

\newcommand{\TheAuthors}{Matthew M.\ Dunlop and Georg Stadler}

\headers{Gradient-free subspace-adjusting ensemble MCMC}{\TheAuthors}

\title{{\TheTitle}}

\author{
  Matthew M. Dunlop\thanks{Courant Institute of Mathematical Sciences, New York University, New York, New York, 10012, USA (\email{matt.dunlop@nyu.edu, stadler@cims.nyu.edu})}\and Georg Stadler\footnotemark[1].
}

\ifpdf
\hypersetup{
  pdftitle={\TheTitle},
  pdfauthor={\TheAuthors}}
\fi

\begin{document}

\maketitle

\begin{abstract}
Sampling of sharp posteriors in high dimensions is a challenging problem, especially when gradients of the likelihood are unavailable. In low to moderate dimensions, affine-invariant  methods, a class of ensemble-based gradient-free methods, have found success in sampling concentrated posteriors. However, the number of ensemble members must exceed the dimension of the unknown state in order for the correct distribution to be targeted. Conversely, the preconditioned Crank-Nicolson (pCN) algorithm succeeds at sampling in high dimensions, but samples become highly correlated when the posterior differs significantly from the prior. In this article we combine the above methods in two different ways as an attempt to find a compromise. The first method involves inflating the proposal covariance in pCN with that of the current ensemble, whilst the second performs approximately affine-invariant steps on a continually adapting low-dimensional subspace, while using pCN on its orthogonal complement.

\end{abstract}



\section{Introduction}
Over the last decade, solving Bayesian inverse problems with high-dimensional parameters has become increasingly feasible due to growing computational resources and the development of methods that scale effectively with the dimension of the state space. To characterize the solution of a Bayesian inference problem, which is the posterior distribution of the parameters, one typically relies on Markov chain Monte Carlo (MCMC) sampling methods.

Abstractly, one is interested in sampling a probability measure $\mu$ that is absolutely continuous with respect to a simpler measure $\mu_0$,
\[
\mu(\dee u) \propto \exp(-\Phi(u))\,\mu_0(\dee u),
\]
where $\Phi(\cdot)$ is a the negative log likelihood, a typically costly-to-evaluate function. Ensemble sampling methods use a set of particles to estimate properties of $\mu$ that can inform the proposal step in MCMC. A particular class of ensemble methods are affine-invariant ensemble methods, whose behavior is invariant under affine transformations.
The nature of affine-invariant sampling methods means that the number of particles required must exceed the dimension of the state $u$, or else they will only sample the distribution restricted to the span of the initial ensemble. \Cref{fig:3d_example} illustrates this effect for a toy example of sampling a standard normal distribution in three dimensions using three ensemble members: the initial ensemble state defines a plane, and the MCMC chains are unable to leave this plane. Additionally, even if a sufficient number of particles are used, the correct distribution may not be targeted; the paper \cite{huijser2015properties} investigates this in the case that the target distribution is a high-dimensional Gaussian. Generally, the performance of ensemble samplers is known to degrade in higher dimensions.

Dimension-robust sampling methods are methods whose performance does not degrade for increasing dimension. However, they are are known to converge slowly when $\mu$ differs substantially from the reference measure $\mu_0$. The aim of this paper is to develop a hybrid version of these two opposite-end sampling approaches, i.e., benefit from the convergence properties of ensemble samplers for concentrated distributions $\mu$ while avoiding degeneration of sampling performance in high dimensions.

\begin{figure}
    \centering
    \includegraphics[width=\textwidth]{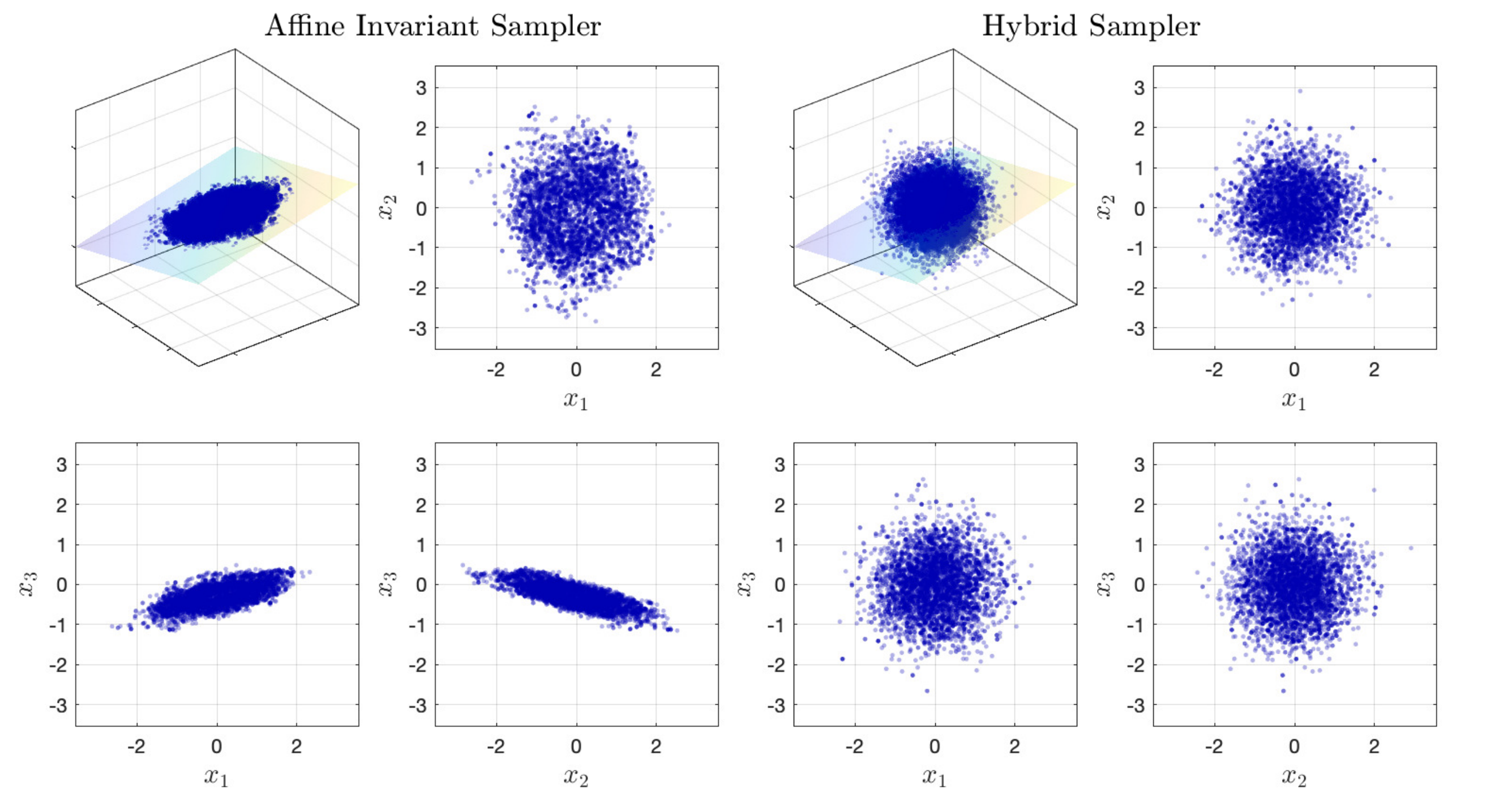}\vspace{-3ex}
    \caption{An example of the affine-invariant and hybrid samplers for the distribution $N(0,I)$ in $\R^3$, with 5000 samples shown for each sampler and an ensemble consisting of three particles. The planes indicate the span of the initial ensemble of particles; the affine-invariant sampler cannot leave this plane while the particles of the hybrid samplers are not restricted to the plane.}
    \label{fig:3d_example}
\end{figure}

\subsection{Related work}
Sampling of high-dimensional posterior distributions via MCMC has received much attention in the past decade, utilizing the formulation of the Metropolis-Hastings algorithm on general state spaces introduced in \cite{tierney1994markov}. The preconditioned Crank-Nicolson (pCN) method  \cite{cotter2013mcmc} is a simple gradient-free sampler, in the case of Gaussian priors, with the property that its convergence rate is bounded independently of the dimension of the state space \cite{hairer2014spectral}. Modifications of pCN are considered in \cite{pinski2015algorithms, rudolf2018generalization} wherein proposals may be more adapted to the posterior distribution, typically using derivative information of the likelihood. The paper \cite{vollmer2015dimension} provides a framework for constructing analogous samplers when the prior is non-Gaussian. Other samplers with dimension-independent convergence properties, utilizing derivative information, include $\infty$-MALA, $\infty$-HMC and their manifold variants \cite{BeskosGirolamiLanEtAl17} and DILI \cite{CuiLawMarzouk16}. The papers \cite{KimVillaParnoEtAl21, BeskosGirolamiLanEtAl17} provide a systematic comparison of a number of the above algorithms applied to high-dimensional Bayesian inverse problems. Outside of MCMC, \cite{agapiou2017importance} considers the performance of importance sampling on general state spaces, and its dependence on the discretization dimension and effective dimension of the problem. Variations of the ensemble Kalman Filter (EnKF) \cite{evensen2003ensemble} have also been considered in the context of Bayesian inversion on general state spaces \cite{garbuno2020interacting,pavliotis2022derivative}, allowing for derivative-free exploration of the posterior via approximate natural Langevin dynamics \cite{garbuno2020interacting}.

When the dimension of the state space is finite and relatively low, affine-invariant ensemble samplers (AIES) \cite{GoodmanWeare10,foreman2013emcee} can allow for efficient derivative-free exploration of complicated posterior distributions; in particular those that are highly concentrated due to particularly informative observations. The continuous time limit of one such algorithm has been studied \cite{Garbuno-InigoNuskenReich20}, resulting in certain Langevin dynamics. The dependence of affine-invariant samplers on dimension has also been studied \cite{huijser2015properties}, showing certain degeneration as the dimension increases. To help overcome the dimensional limitations of affine-invariant samplers, \cite{coullon2021ensemble} introduces a sampler that uses AIES on a subspace defined using the prior distribution, and pCN on its complement. In this article we take a similar approach wherein the subspace is not as strongly constrained by the prior, which can be more appropriate in the setting of concentrated posteriors.

\subsection{Contributions and limitations}
Our main contributions are as follows:
(1) We propose two gradient-free ensemble sampling algorithms that are well-defined in infinite dimensions. In these methods, the  
subspace spanned by the ensemble is not fixed, and thus there is no minimally required ensemble size.
(2) We numerically study the new methods' performance for different ensemble size, and compare their performance to existing methods for linear, nonlinear and non-smooth infinite-dimensional Bayesian inverse problems.

The proposed methods also have limitations: 
(1) Since the ensemble is used to compute a Gaussian proposal distribution in the subspace spanned by the particles, the method loses efficiency for strongly non-Gaussian densities.
(2) Our algorithms require some parameter choices, e.g., the ensemble size, a jump parameter in MCMC, and the dimension of a subspace in one of the methods. However, we will show numerically that the algorithms' performance is rather insensitive to these choices.

\section{Bayesian inverse problems}
In this section we provide an overview of the Bayesian approach to inverse problems, combining the observation model with the prior measure to construct the posterior measure on general state spaces. We then discuss the problem of producing samples from the posterior numerically and various issues that may arise.

\subsection{The prior, likelihood and posterior}
Suppose that we have data $y \in Y$ arising from some nonlinear noisy observations of a state $u \in X$, and our goal is to estimate $u$ from $y$. We write 
\[
y = F(\G(u),\eta)
\]
for some forward map $\G:X\to Y_1$, random noise $\eta \in Y_2$ and state-to-observation map $F:Y_1\times Y_2\to Y$. A common setup is that of additive Gaussian noise: $Y_1 = Y_2 = Y$, $F(z,\eta) = z+\eta$ and $\eta \sim N(0,\Gamma)$, so that
\[
y = \G(u) + \eta,\quad \eta\sim N(0,\Gamma).
\]
Such problems are typically ill-posed from a classical perspective: there may exist no solution, the solution may not be unique, or the solution may be highly sensitive to the realization of the noise $\eta$. In this article, we consider the underdetermined case wherein that $X$ is high- or infinite-dimensional Hilbert space and $Y = \R^J$ is finite-dimensional.

We consider the Bayesian approach to the inversion wherein rather than a single state $u \in X$ as a solution, we seek a probability distribution $\mu$ on $X$. If we quantify our prior beliefs about unknown $u$ by a measure $\mu_0(\dee u) = \P(\dee u)$ on $X$ and provide a probability distribution for the noise $\eta$, this induces a likelihood function $\P(y|u)$. For example in case of additive Gaussian noise above, the likelihood is given by
\[
\P(y|u) \propto \exp\left(-\frac{1}{2}\|\cG(u)-y\|_\Gamma^2\right),
\]
where $\|\cdot\|_\Gamma = \|\Gamma^{-\frac{1}{2}}\cdot\|$. We define the solution to the Bayesian inverse problem as the measure $\mu(\dee u) = \P(\dee u|y)$, where by Bayes' theorem
\[
\P(\dee u|y) = \frac{\P(y|u)\P(\dee u)}{\P(y)}.
\]
To be more explicit, suppose that the likelihood takes the form
\[
\P(y|u) \propto \exp(-\Phi(u;y)),
\]
where $\Phi$, referred to as the negative log-likelihood, is sufficiently regular \cite{DashtiStuart16}. Then the Bayesian posterior $\mu$ is absolutely continuous with respect to the prior $\mu_0$, and its Radon-Nikodym derivative takes the form\footnote{In what follows we drop the dependence of $\Phi$ on the data $y$ as we assume it fixed.}
\begin{align}
\label{eq:bayes}
\frac{\dee \mu}{\dee \mu_0}(u) = \frac{1}{Z}\exp(-\Phi(u)),\quad Z = \int_X \exp(-\Phi(u))\,\mu_0(\dee u).
\end{align}
Though our motivation is Bayesian inversion, the methodology introduced in this article may be used to sample general measures with the form \cref{eq:bayes}, for example Gibbs measures. We will still however refer to $\Phi$ as the negative log-likelihood and $\mu_0$ as the prior.

In this article, we focus on the case where the prior measure $\mu_0 = N(0,C_0)$ is a centered Gaussian. In this case, when the dimension of $X$ is finite,  the posterior admits a Lebesgue density $\mu(\dee u) = \pi(u)\,\dee u$,
\[
\pi(u) \propto \exp\left(-\Phi(u) - \frac{1}{2}\|u\|_{C_0}^2\right),
\]
which will be useful to consider to provide a formal interpretation of the infinite-dimensional algorithms introduced in the following section.

\begin{remark}
\label{rem:noncenter}
The assumption that the prior measure $\mu_0$ is a centred Gaussian is not as strong as it first appears. For example, suppose that the prior takes the form
\[
\mu_0(\dee u) =  \frac{1}{Z_0}\left(-\Phi_0(u)\right)\,\nu_0(\dee u),\quad Z_0 = \int_X \exp(-\Phi_0(u))\,\nu_0(\dee u),
\]
where $\nu_0$ is the pushforward of a Gaussian measure $\nu_r$ via a possibly non-linear map $T:\Xi\to X$, i.e., $\nu_0 = T^\sharp \nu_r$. Then $\mu = T^\sharp \nu$, where
\[
\frac{\dee \nu}{\dee \nu_r}(\xi) = \frac{1}{Z_r}\exp\left(-\Psi(\xi)\right),\quad Z_r = \int_\Xi \exp(-\Psi(\xi))\,\nu_r(\dee \xi)
\] 
and $\Psi(\xi) = \Phi(T(\xi)) + \Phi_0(T(\xi))$. The measure $\nu$ is precisely of the form \cref{eq:bayes} with a Gaussian dominating measure and so if we can sample $\nu$, we can sample $\mu$ by transforming the samples with $T$. The simple case $T(\xi) = \xi+m_0$ and $\Phi_0 = 0$ illustrates why it is sufficient to assume the dominating Gaussian is centred, for example. In the algorithms we consider, it can be useful to work with a white noise measure $\nu_r = N(0,I)$, $\Phi_0 = 0$ and define $T(\xi) = C_0^{{1}/{2}}\xi$ to map to the prior Gaussian measure $\mu_0 = \nu_0$.
\end{remark}

\subsection{Probing the posterior}

Though the solution $\mu$ exists abstractly as a measure under relatively mild assumptions on the prior and negative log-likelihood, one is often interested in getting information from this measure numerically. For example, one may desire the mode, the mean or estimates and confidence bounds on quantities of interest. The latter typically require samples from the posterior to estimate via Monte Carlo, since the integrals involved are often high-dimensional. Producing these samples can be challenging in many setups, for example,
\begin{enumerate}[(i)]
\item when the data is particularly informative, the posterior distribution can be concentrated on a lower dimensional submanifold of $X$ which needs to be discovered;
\item effective sampling methods often make use of derivatives of the posterior density, but these may not exist, may be unknown, or may be computationally prohibitive to evaluate;
\item the posterior may have multiple distinct modes, which many sampling methods may struggle to explore; and
\item when the dimension of the space $X$ is infinite, the posterior cannot admit a Lebesgue density $\pi$, however many sampling algorithms are defined in terms of a Lebesgue density. Similarly in high but finite dimensions, the posterior is often almost singular with respect to the Lebesgue measure, leading to statistical issues with said algorithms.
\end{enumerate}

In the remainder of this article we consider certain Markov chain Monte Carlo (MCMC) methods with the aim of partly resolving the above points.

\section{MCMC sampling for Bayesian inverse problems}
\label{sec:mcmc}
In this section we first give an overview of Metropolis-Hastings MCMC algorithms, outline two classes of such algorithms (affine-invariant and dimension-robust) and describe some of their respective advantages and disadvantages. We then introduce two hybrid methods that interpolate between the two classes as an approach to ameliorating some of their disadvantages.

\subsection{Metropolis-Hastings MCMC sampling}
\label{ssec:metropolis}
MCMC methods aim to sample a given probability distribution by constructing a Markov chain for which it is the stationary distribution. A common construction of such a chain is via a Metropolis-Hastings propose-accept-reject mechanism. Given a target probability distribution $\pi$ and a state $u_k$, a new state $\hat{u}_{k}$ is proposed according to a proposal distribution $\hat{u}_{k} \sim q(u_k,u)\,\dee u$. One then sets $u_{k+1} = \hat{u}_{k}$ with probability
\[
\min\left\{1,\frac{\pi(\hat{u}_k)q(u_k,\hat{u}_k))}{\pi(u_k)q(\hat{u}_k,u_k))}\right\},
\]
or else sets $u_{k+1} = u_k$. This choice of acceptance probability ensures that the resulting Markov chain satisfies detailed balance and hence has the desired stationary distribution. A simple choice of proposal distribution is a symmetric random walk proposal, $q(u_k,\cdot) = N(u_k,C)$ for some jump covariance $C$, in which case the algorithm is referred to as Random Walk Metropolis (RWM). However, depending on the structure of the target measure $\pi$, a more complex proposal distribution is typically more efficient computationally. Once the Markov chain has reached stationarity (after a period referred to as burn-in), the samples $\{u_k\}$ may be used to approximate quantities of interest, such as mean, variance, or marginal probability distributions. The samples $\{u_k\}$ are typically correlated -- ideally one wishes to produce a Markov chain whose samples are as least correlated as possible in order to estimate these quantities of interest efficiently.

\subsection{Affine-invariant MCMC sampling}
\label{ssec:affine}

Instead of targeting the posterior density $\pi$ on $X$ directly, an ensemble of particles is used to target the product measure
\begin{align}
\label{eq:pistar}
\pi^*(\dee u^{(1)},\ldots,\dee u^{(N)}) = \prod_{j=1}^N \pi(\dee u^{(j)})
\end{align}
on $X^N$. This immediately provides two advantages over a single chain targeting $\pi$:
\begin{enumerate}[(i)]
\item if the posterior is multimodal, different particles can explore distinct modes without the need to move between them; and 
\item the empirical distribution of the ensemble at a given step provides a coarse estimate for the posterior distribution, which can be used to adapt the proposal distribution.
\end{enumerate}
A class of ensemble methods, called affine-invariant methods, were introduced in \cite{GoodmanWeare10}. Suppose that the MCMC update for a particular particle takes the form
\[
u^{(j+1)} = R(u^{(j)},\xi^{(n)}),
\]
where $\xi^{(n)}$ is a random variable. We say that the update is affine-invariant if for any $A \in \mathcal{L}(X,X)$ and $b \in X$,
\[
R(Au+b,\xi) = AR(u,\xi) + b.
\]
If a method has this property then as a consequence, distributions which are concentrated around a hyperplane are as easy to sample as those which are more dispersed; see \cite{GoodmanWeare10} for more details.
We provide an overview of an example of an affine-invariant method introduced in \cite{GoodmanWeare10}, referred to as the \emph{walk move}. Given a particle $u^{(j)}$ we denote $u^{(-j)}$ the complementary ensemble
\[
u^{(-j)} = \{u^{(1)},\ldots,u^{(j-1)},u^{(j+1)},\ldots,u^{(N)}\}
\]
Then given a subcollection $S\subseteq u^{(-j)}$, after the chain has reached stationarity, the sample covariance of $S$ should provide an approximation to posterior covariance. Thus, one can perform RWM updates where the proposal covariance is proportional to this sample covariance. The algorithm is given explicitly in \cref{alg:affine}, and referred to as the Affine Invariant Ensemble Sampler (AIES).

\begin{algorithm}
\begin{algorithmic}
\caption{Affine-invariant MCMC sampling}
\label{alg:affine}
\State Choose initial ensemble of particles $\{u_1^{(n)}\}_{n=1}^N \subseteq X$ and jump parameter $\lambda > 0$.
\For{$k=1:K$}
\For{$n=1:N$}
\State Choose $S \subseteq u_k^{(-n)}$ and propose
\[
\hat{u}_k^{(n)} = u_k^{(n)} + \lambda\cdot\frac{1}{\sqrt{|S|}}\sum_{u_k^{(j)} \in S} z_j(u_k^{(j)} - \overline{u}_S),\quad z_j \iid N(0,1).
\]
\State Set $u_{k+1}^{(n)} = \hat{u}_k^{(n)}$ with probability
\[
\min\left\{1,\exp\left(\Phi(u_k^{(n)}) - \Phi(\hat{u}_k^{(n)}) + \frac{1}{2}\|u_k^{(n)}\|_{C_0}^2 - \frac{1}{2}\|\hat{u}_k^{(n)}\|_{C_0}^2\right)\right\}
\]
\State or else set $u_{k+1}^{(n)} = u_k^{(n)}$.
\EndFor
\EndFor
\State\Return $\{u_k^{(n)}\}_{k,n=1}^{K,N}$.
\end{algorithmic}
\end{algorithm}

Other affine invariant proposals are available, such as the stretch move \cite{GoodmanWeare10}. However, these proposals often have a strong dimensional dependence, and in particular are not well-defined in infinite dimensions. Recently the ALDI method has been introduced \cite{Garbuno-InigoNuskenReich20}, which involves simulating an appropriate affine-invariant Langevin diffusion targeting $\pi^*$ with an Euler-Maruyama scheme; this is a modification of the Ensemble Kalman Sampler \cite{garbuno2020interacting} such that the correct distribution is targeted.
A drawback of the above methods is that in order for them to sample the correct distribution, the number of particles $N$ must be larger than the dimension of the state space $X$. For example, when using \cref{alg:affine} the particles cannot move out of the lowest dimension hyperplane passing through the initial ensemble; see \cref{fig:3d_example} for a simple illustration in three dimensions. When the dimension of the state space is high or infinite this requirement can make the algorithm impractical.

\subsection{Dimension-robust MCMC sampling}
\label{ssec:robustmcmc}
As we are interested in the case when $X$ is high- or infinite-dimensional, we ideally desire a sampling method that is well-defined in infinite dimensions to bypass dimension-dependent issues. Such methods have received much attention recently, though the general Metropolis-Hastings algorithm was formulated on Hilbert space in 1994 \cite{tierney1994markov}. Key to the construction of these algorithms is that the posterior is absolutely continuous with respect to a dominating measure -- in our setup we assume this to be Gaussian, rather than the Lebesgue measure as is typically the case in finite dimensions. The notion of dimension-robustness informally refers to the algorithm being well-defined and ergodic on Hilbert space, and more rigorously defined as the geometric rate of convergence to stationarity with respect to some metric on measures being bounded below by some positive constant independently of dimension.

An example of a dimension-robust MCMC method, assuming a Gaussian prior, is the preconditioned Crank-Nicolson (pCN) method \cite{cotter2013mcmc}. This is a modification of the random walk Metropolis algorithm such that for the proposal, the current state is rescaled and then perturbed by a Gaussian random variable with covariance proportional to the prior covariance:
\[
\hat{u}_k = \sqrt{1-\beta^2}u_k + \beta\xi,\quad \xi\sim N(0,C_0),
\]
for some $\beta \in (0,1]$. The acceptance probability is then simply a likelihood ratio -- the prior information is fully contained in the proposal. When a more general prior mean is $m_0 \in X$ is assumed, the proposal
\[
\hat{u}_k = m_0 + \sqrt{1-\beta^2}(u_k-m_0) + \beta\xi,\quad \xi\sim N(0,C_0),
\]
is instead used, with the same acceptance probability. Whilst this algorithm works in arbitrarily high dimensions when the posterior is absolutely continuous with respect to the prior, its performance in terms of mixing can be poor when the posterior is far from the prior, i.e., when the data is particularly informative and the likelihood is very skewed: in order to maintain a reasonable acceptance rate, the parameter $\beta$ must be chosen extremely small, and so samples are highly correlated. One approach is to, instead of using jumps based on the prior, use jumps from some other Gaussian distribution that has been informed by the likelihood. As long as the jump distribution is equivalent to the prior, the modification to the acceptance probability to ensure detailed balance holds is well-defined in infinite dimensions. Specifically, suppose that the jump distribution is taken to be $N(m,C)$, then the proposal distribution is given by
\begin{align*}
Q(u,dv) &= N\left(m + \sqrt{1-\beta^2}(u-m),\beta^2 C\right).
\end{align*}
Defining the measures $\omega$, $\omega^\top$ on the product space $X\times X$ by
\[
\omega(\dee u,\dee v) = \mu(\dee u)Q(u,\dee v),\quad \omega^\top(\dee u,\dee v) = \omega(\dee v,\dee u),
\] 
following \cite{tierney1994markov} the acceptance probability is then given by
\[
\alpha(u,\hat{u}) = \min\left\{1,\frac{\dee \omega}{\dee \omega^\top}(\hat{u},u)\right\},
\]
which is well-defined by the absolute continuity $\mu\ll\mu_0$ and assumed equivalence of the prior and jump distributions. The algorithm, referred to as generalized pCN (gpCN), is given in \cref{alg:gpcn} after calculating this Radon-Nikodym derivative\footnote{A related algorithm, introduced in \cite{rudolf2018generalization}, is also referred to as gpCN. This algorithm also modifies the jump distribution of standard pCN away from the prior, except its mean is modified in such a way that the expression for the acceptance probability remains the same.}. To the authors' knowledge, this algorithm first appeared in \cite{pinski2015algorithms}. If derivatives of the likelihood are available, a typical example of jump distribution one can use is the Laplace approximation to the posterior or some approximation thereof \cite{bui2013computational,pinski2015algorithms}. In the following subsection we consider a different class of jump distributions which do not necessarily require derivatives. The case $m=0$ and $C = C_0$ provides the original pCN algorithm, in which case $I_C \equiv 0$. The parameters $\beta$, $m$ and $C$ may be chosen to depend on the time step $k$, for example if one were to use an adaptive variant of the above; however conditions on the dependence on $k$ are required in order to preserve ergodicity.

\begin{algorithm}
\begin{algorithmic}
\caption{Generalized pCN MCMC sampling}
\label{alg:gpcn}
\State Choose initial state $u_1 \in X$.
\For{$k=1:K$}
\State Propose
\[
\hat{u}_k = m + \sqrt{1-\beta^2}(u_k-m) + \beta \xi,\quad \xi \sim N(0,C)
\]
\State Set $u_{k+1} = \hat{u}_k$ with probability
\[
\min\left\{1,\exp\left(\Phi(u_k) - \Phi(\hat{u}_k) + I_C(u_k) - I_C(\hat{u}_k)\right)\right\}
\]
\State where
\begin{align}
\label{eq:IC}
I_C(u) = \frac{1}{2}\Big\langle u,\big(I-C_0^{\frac{1}{2}}C^{-1}C_0^{\frac{1}{2}}\big)u\Big\rangle_{C_0} - \langle u,m\rangle_C + \frac{1}{2}\|m\|_C^2,
\end{align}
\State or else set $u_{k+1} = u_k$.
\EndFor
\State\Return $\{u_k\}_{k=1}^K$.
\end{algorithmic}
\end{algorithm}

Note that in finite dimensions, $I_C(u)$ can be defined more directly as 
\[
I_C(u) = \frac{1}{2}\|u\|_{C_0}^2 - \frac{1}{2}\|u-m\|_C^2.
\]
In infinite dimensions, however, each of these terms is infinite almost surely. In the definition of $I_C$ in \cref{eq:IC} each term is finite, see \cref{appendix}.

\subsection{A hybrid algorithm}
\label{ssec:hybrid}
The affine-invariant MCMC has the advantage of adapting well to skewed posterior distributions. However, it targets the incorrect distribution if there are fewer particles than dimensions, making it impractical for high-dimensional problems. Conversely the gpCN algorithm is well-defined in infinite dimensions, but performs poorly if the posterior is far from the prior (or chosen jump distribution). We balance these issues by interpolating between the two algorithms.

We return to the setup of the affine-invariant MCMC algorithm, and target the product measure $\pi^*$ on $X^N$ given by \cref{eq:pistar}. Let $\R^D$ represent a discretization of $X$. Given a set of particles $S$, of size $|S|$, we define the normalized centred data matrix $V_S \in \R^{D\times |S|}$ by
\[
(V_S)_{:, j} := \frac{1}{\sqrt{|S|-1}}\left(u^{(j)} - \overline{u}_S\right),
\]
where $\overline{u}_S = \frac{1}{|S|}\sum_{u^{(j)} \in S} u^{(j)}$ is the sample mean of the particles in $S$.
The matrix $V_SV_S^\top$ then provides the sample covariance of the particles $S$.

In the inner loop of the affine-invariant MCMC algorithm, we use gpCN with the choice of jump distribution $N(m,C)$ with $C = C_* + \gamma^2V_SV_S^\top$ for some $\gamma > 0$, where the measures $N(0,C_0)$ and $N(0,C_*)$ are assumed equivalent. This jump distribution possesses the requisite prior equivalence due to the following simple proposition. 
\begin{proposition}
\label{prop:hybrid_equiv}
Let $C_* \in \mathcal{L}(X;X)$ be positive and trace-class, and let $W \in \mathcal{L}(X;X)$ have finite rank. Then the measures $N(0,C_*)$ and $N(0,C_*+W)$ are equivalent.
\end{proposition}

\begin{proof}
By the Feldman-Hajek theorem \cite{da2002second}, it is sufficient to show that the operator $C_*^{-\frac{1}{2}}(C_*+W)C_*^{-\frac{1}{2}}-I = C_*^{-\frac{1}{2}}WC_*^{-\frac{1}{2}}$ is Hilbert-Schmidt; this follows immediately since $W$ has finite rank. 
\end{proof}

The sample covariance $V_SV_S^\top$ has rank at most $|S|-1$; indeed its range is the lowest dimensional hyperplane passing through the elements of $S$, shifted to intersect the origin. By the assumed equivalence of $N(0,C_0)$ and $N(0,C_*)$ the equivalence of the prior and jump distributions follows.

Given a gpCN jump parameter $\beta \in (0,1]$ and RWM jump parameter $\lambda > 0$, we make the choice $\gamma = \lambda/\beta$ so that the proposal distribution is given by
\[
Q(u_k^{(j)},\cdot) = N\left(m + \sqrt{1-\beta^2}(u_k^{(j)}-m), \beta^2\left(C_* + \frac{\lambda^2}{\beta^2}V_SV_S^\top\right)\right).
\]
Hence, as $\beta \to 0$ we recover the affine-invariant proposal as in \cref{alg:affine}, and as $\lambda\to 0$ we recover the gpCN algorithm for each particle . Note that, for $\beta > 0$, the proposals are not restricted to a hyperplane dictated by the initial ensemble. The full algorithm is given in \cref{alg:hybrid}. We refer to this algorithm as \hybrid (\hybridname).

\begin{algorithm}
\begin{algorithmic}
\caption{\hybrid MCMC sampling}
\label{alg:hybrid}
\State Choose initial ensemble of particles $\{u_1^{(n)}\}_{n=1}^N \subseteq X$ and jump parameters $\beta \in (0,1]$, $\lambda > 0$. Set $\gamma = \lambda/\beta$.
\For{$k=1:K$}
\For{$n=1:N$}
\State Choose $S \subseteq u_k^{(-n)}$ and propose
\begin{align*}
\hat{u}_k^{(n)} &= m + \sqrt{1-\beta^2}(u_k^{(n)}-m) + \beta \xi + \lambda\sum_{j \in S} (V_S)_{:, j}z_j,\\
&\quad z_j\iid N(0,1),\quad\xi \sim N(0,C_*).
\end{align*}
\State Set $u^{(n)}_{k+1} = \hat{u}_k^{(n)}$ with probability
\[
\min\left\{1,\exp\left(\Phi(u_k^{(n)}) - \Phi(\hat{u}_k^{(n)}) + I_{C}(u_k^{(n)}) - I_{C}(\hat{u}_k^{(n)})\right)\right\}
\]
\State where $C = C_*+\gamma^2 V_SV_S^\top$, or else set $u_{k+1}^{(n)} = u_k^{(n)}$.
\EndFor
\EndFor
\State\Return $\{u^{(n)}_k\}_{k,n=1}^{K,N}$.
\end{algorithmic}
\end{algorithm}

\begin{remark}
\begin{enumerate}[(i)]
\item A special case of \cref{alg:hybrid} is $m = 0$ and $C_* = C_0$. Denote $P_0 = C_0^{-1}$ the prior precision operator, noting that this is often a local operator and hence sparse when implemented numerically. In this case, $I_{C}$ reduces to
\begin{align}
I_{C}(u) = \frac{1}{2}\left\langle V_S^\top P_0 u, (\gamma^{-2}I + V_S^\top P_0 V_S)^{-1} V_S^\top P_0 u \right\rangle_X,
\end{align}
where the matrix being inverted is small $(N\times N)$. 

\item For $\beta > 0$ this algorithm is \emph{not} affine-invariant. However, it is approximately affine-invariant for small $\beta$. Writing $\hat{u} = R(u,\xi)$ for the proposal, as in \cref{ssec:affine}, we have
\begin{align*}
R(Au+b,\xi) &= AR(u,\xi) + b + \beta(I-A)\xi + (\sqrt{1-\beta^2}-1)b\\
&= AR(u,\xi) + b + \mathcal{O}(\beta).
\end{align*}

\item There are two jump parameters that may be tuned in the \hybrid algorithm: $\beta$ corresponding to the pCN jump size and $\lambda$ corresponding to the size of the prior perturbation.  Jointly finding the optimal values of these parameters can be difficult in practice; in all numerical examples in this article we simply fix $\lambda = 0.2$ and adapt $\beta$ so that the acceptance rate lies in the interval $(0.15,0.3)$, which appears to be effective empirically. In practice we find starting $\beta$ at a large value is beneficial, allowing for the initial ensemble to adapt to find the effective support of the posterior, before it is reduced to allow for the neighbourhood of the corresponding hyperplane to be explored without too many rejections.

\item In order to provide shift invariance of the proposal it could be tempting to make the choice $m = \overline{u}_S$, the sample mean. However, for a finite number of particles, $\overline{u}_S$ does not lie in the Cameron-Martin space of the prior, and so the acceptance probability is not be well-defined due to measure singularity. One could however consider $m = P\overline{u}_S$ for some projection $P$ onto the Cameron-Martin space.

\item Though one has free choice over the subset $S\subseteq u^{(-n)}$ used to estimate the covariance, we found the choice $S = u^{(-n)}$ to be effective in practice. However, choosing $S$ to be a proper subset of $u^{(-n)}$ may be beneficial when the posterior is multimodal, as well as providing robustness with respect to outliers during burn-in. In the case of multiple separated modes, our sampler would have to additionally be combined with a method that allows samples to switching between modes to accurately measure the relative importance of individual modes \cite{LindseyWeareZhang22,GabrieRotskoffVandeneijnden21}.

\end{enumerate}
\end{remark}

We note that although the dimension $D$ of an inverse problem may be large, often the effective dimension of the problem is much smaller -- the posterior may be concentrated on some low-dimension submanifold of $X$, relative to the prior. It is for this reason that we expect the above algorithm to remain effective when $D$ is large for a finite number of particles. The paper \cite{agapiou2017importance} introduces a quantitative notion of effective dimension for linear Gaussian Bayesian inverse problems, defined in terms of the prior-weighted Gauss-Newton Hessian $Q$. The operator $Q$ may be used to estimate the dimension of the subspace that is informed by the data, relative to the prior. This dimension then gives a rough indication for the order of magnitude of number of particles that should be used in order to achieve good mixing with the \hybrid algorithm. In the notation of \cref{ssec:linear_numerics}, the effective dimension $\mathsf{efd} \leq D$ is defined as\footnote{The paper \cite{agapiou2017importance} also considers an alternative definition of effective dimension simply given by $\tr(Q)$, however this is not bounded above by the dimension of the state space $X$.}
\[
\mathsf{efd} = \tr(Q(I+Q)^{-1}),\quad Q = C_0^{\frac{1}{2}}A^*\Gamma^{-1}AC_0^{-\frac{1}{2}}.
\]

\subsection{An alternative hybrid algorithm}

The paper \cite{coullon2021ensemble} introduces an algorithm that combines affine-invariant sampling with pCN, the Functional Ensemble Sampler (FES), wherein an affine-invariant method is applied on a subspace defined via the prior distribution and pCN is applied on the complementary subspace. Specifically, the affine-invariant method is applied on the subspace spanned by the first $M$ modes of the Karhunen-Lo\'eve expansion of the prior. Thus, if the prior Gaussian distribution on $u$ has Karhune-Lo\'eve expansion
\[
u = \sum_{j=1}^\infty \sqrt{\lambda_j}\xi_j\varphi_j,\quad \xi_j \iid N(0,1),
\]
then a Gibbs-type MCMC algorithm is used to perform affine-invariant updates on the components $\xi_1,\ldots,\xi_M$, and pCN on the remaining components $\xi_{M+1},\ldots$. This is effective when the posterior is relatively close to the prior, however when the data is particularly informative, these prior modes do not represent the posterior well and performance is closer to plain pCN. We consider an adjustment of this algorithm, in the spirit of the hybrid algorithm introduced above, which adapts the subspace based upon the current ensemble. That is, given a subspace dimension $M$ and an ensemble $S$ we diagonalize the sample covariance $V_S V_S^\top$, and truncate this expansion after the first $M$ singular vectors. We then effectively perform an (approximately) affine-invariant update on the span of the first $M$ singular vectors of the sample covariance, and pCN on the orthogonal complement. The specific algorithm is given in \cref{alg:hybrid_proj}, and we refer to this as the \hybridp (\hybridpname) algorithm; for convenience we assume the prior is white as discussed in \cref{rem:noncenter} to avoid the requirement for simultaneous diagonalization of the prior and sample covariances. This method almost agrees with the \hybrid method introduced above, except the contributions to the proposal covariance arising from the sample covariance and the prior are performed on orthogonal subspaces -- again this corresponds to a low-rank update of the prior, and so we may use \cref{prop:hybrid_equiv} to see that the algorithm is well-defined. We compare the \hybrid and \hybridp algorithms with pCN and FES in the following section.

\begin{algorithm}
\begin{algorithmic}
\caption{\hybridp MCMC sampling}
\label{alg:hybrid_proj}
\State Choose initial ensemble of particles $\{u_1^{(n)}\}_{n=1}^N \subseteq X$, jump parameters $\beta \in (0,1]$, $\lambda > 0$, and subspace dimension $M$. Set $\gamma = \lambda/\beta$.
\For{$k=1:K$}
\For{$n=1:N$}
\State Choose $S \subseteq u^{(-n)}$ and let $U\Sigma U^\top = V_SV_S^\top$ be the SVD of the covariance of $S$. 
\State Denote $\Sigma_M$ the restriction of $\Sigma$ to the $M$ largest singular values, and $U_M$ the 
\State corresponding columns of $U$. Propose
\begin{align*}
\hat{u}_k^{(n)} &= m + \sqrt{1-\beta^2}(u_k^{(n)}-m) + \beta \left[U_M \Big(\gamma\Sigma_M^{\frac{1}{2}} - I\Big)U_M^\top\xi + \xi\right],\quad \xi \sim N(0,I).
\end{align*}
\State Set $u_{k+1}^{(n)} = \hat{u}_k^{(n)}$ with probability
\[
\min\left\{1,\exp\left(\Phi(C_0^{\frac{1}{2}}u_k^{(n)}) - \Phi(C_0^{\frac{1}{2}}\hat{u}_k^{(n)}) + J(U_M^\top u_k^{(n)}) - J(U_M^\top\hat{u}_k^{(n)})\right)\right\}
\]
\State where $J(z) = \frac{1}{2}\|z\|_{\R^M}^2 - \frac{1}{2\gamma^2}\left\langle z,\Sigma_M^{-1}z\right\rangle_{\R^M}$, or else set $u_{k+1}^{(n)} = u_k^{(n)}$.
\EndFor
\EndFor
\State\Return $\{C_0^{\frac{1}{2}}u^{(n)}_k\}_{k,n=1}^{K,N}$.
\end{algorithmic}
\end{algorithm}

\section{Numerical illustrations}
\label{sec:numerics}

We numerically compare the behaviour of the pCN, FES and hybrid algorithms for three different inverse problems. We first consider a linear inverse problem, to investigate the effect of the number of particles and the dimension of the problem. We then consider a nonlinear problem, based on the setup of \cite{Garbuno-InigoNuskenReich20} as well as a generalization, to compare the effect of the sharpness of the posterior distribution relative to the prior on the behavior of the algorithms. Finally we consider a high-dimensional problem with a level-set prior, where no gradients of the likelihood are available.

\subsection{A linear regression problem}
\label{ssec:linear_numerics}
We consider first the case where the forward map is linear, the noise is additive Gaussian, and the prior is Gaussian. In this setup the posterior is Gaussian with a known closed form, and so we can directly compare the result of the sampling with the true posterior in order to assess the accuracy. Let $\Omega = (0,2\pi)$ and define the observation operator $A:C^0(\Omega)\to \R^J$ by $(Av)_j = v(d_j)$. We assume we have data $y \in\ \R^J$  arising from the model
\[
y = Au + \eta,\quad\eta \sim N(0,\gamma^2 I)
\]
for some $\gamma > 0$. The true state $u^\dagger$ generating the data is taken to be $u^\dagger(x) = \sin(x)/2$. We observe the solution at $J = 25$ points, $d_j = 2\pi j/J$, and fix $\gamma = 10^{-3}$ so that the relative $\ell^2$ error on the data is $0.0316\%$. The problem is discretized on a uniform grid of $D$ points.
The prior is taken to be of Mat\'ern type $\mu_0 = N(0,C_0)$, $C_0 = (I-\Delta)^{-1}$, where $\Delta$ is the Laplacian with homogeneous Neumann boundary conditions. The resulting effective dimension of the problem is then approximately $25$. The posterior has closed Gaussian form $\mu = N(m_{\text{p}},C_{\text{p}})$, where
\[
C_{\text{p}}^{-1} = A\Gamma^{-1}A^* + C_0^{-1},\quad C_{\text{p}}^{-1}m_{\text{p}} = A^*\Gamma^{-1}y,
\]
which may be used to assess the accuracy of sampling methods.

\subsubsection{Comparison of algorithms}
We compare the performance of both hybrid algorithms introduced in this paper with the FES algorithm and the pCN algorithm on the above problem. We fix dimension $D=100$ and $N=40$ particles. We do not consider the AIES algorithm, nor the ALDI algorithm of \cite{Garbuno-InigoNuskenReich20}, since we know that these only provide subspace sampling when $N\leq D$. We generate $10^5$ samples per particle so that $4\times 10^6$ samples are generated in total for each method, with the same number of likelihood evaluation required in each case; the first $25\%$ of samples for each particle are discarded as burn-in. For the pCN method, independent chains are run for each particle.

We estimate the autocorrelations for each particle chain and average these over the particles. That is, given samples $\{u_k^{(n)}\}$ and a scalar-valued function $g:X\to\R$, we estimate the function $c:\N\to\R$,
\[
c(j) = \frac{1}{N}\sum_{n=1}^N\frac{c_j^{(n)}}{c_0^{(n)}},\quad c_j^{(n)} = \frac{1}{K}\sum_{k=1}^{K-j}\big(g(u_k^{(n)})-\overline{g}^{(n)}\big)\big(g(u_{k+j}^{(n)})-\overline{g}^{(n)}\big),\quad \overline{g}^{(n)} = \frac{1}{N}\sum_{k=1}^K g(u_k^{(n)}),
\]
where $c_0$ is the sample variance. The area under the graph of $c$ is inversely proportional to the effective number of statistically independent samples in the chain, and so rapid decay of $c$ is desired for an effective sampling algorithm. Throughout this section we will take $g(u) = \|u\|_{L^2}^2$. In \cref{fig:linear_ac_alg} we show the autocorrelations for the four different algorithms. In all cases the jump parameter $\beta$ is adapted so that the acceptance rate lies in the interval $(0.15,0.3)$, and $\lambda = 0.2$ is fixed. In the case of the FES algorithm, the dimension of the subspace AIES is performed upon is chosen as $K=10$, and stretch moves with parameter $a=2$ are performed as suggested in \cite{coullon2021ensemble}; the pCN jump parameter is adapted as above. We see that the autocorrelations decay significantly faster for the algorithms introduced in this paper. The FES method performs similarly to the pCN algorithm as the posterior eigenbasis differs significantly from the prior eigenbasis due to the sharpness of the likelihood. Note that although the autocorrelation for the pCN algorithm decays fast initially due to small scale oscillations, the asymptotic decay of the autocorrelations appears to be faster for the other algorithms. The behavior of the autocorrelations can be further understood from \cref{fig:linear_traces}, which shows the traces of the squared norm of an individual particle for three of the algorithms\footnote{The \hybridp chain is omitted here for brevity; it has the same qualitative behavior as the \hybrid chain.} -- here the long-term correlations for the pCN and FES chains can be observed, and contrasted with the \hybrid chain.

\begin{figure}[bt]
    \centering
    \begin{tikzpicture}[scale=0.8]
	\begin{axis}[
   		width=8cm,
    	height=5cm,
    	title={},
    	xlabel={\large Lag},
    	ylabel={\large Autocorrelation},
    	xmin=0, xmax=5000,
    	ymax=1,
    	ymin=-0.1,
    	ytick={0,0.2,0.4,0.6,0.8,1},
    	grid=major,
    	legend style={draw=none,
        	          at={(1.3, 0.95)},
            	      anchor=north},
    	legend style={font=\large},
    	legend cell align=left]
    	\addplot[pcn,mark=none,style=very thick]  table [x index=0, y index=1, col sep=tab] {data/linear_acs_alg.txt};
    	\addlegendentry{pCN}
    	\addplot[fes,mark=none,style=very thick]  table [x index=0, y index=2, col sep=tab] {data/linear_acs_alg.txt};
    	\addlegendentry{FES}
    	\addplot[hybrid,mark=none,style=very thick]  table [x index=0, y index=3, col sep=tab] {data/linear_acs_alg.txt};
    	\addlegendentry{\hybrid}
    	\addplot[hybrid-proj,mark=none,style=very thick]  table [x index=0, y index=4, col sep=tab] {data/linear_acs_alg.txt};
    	\addlegendentry{\hybridp}
    	\addplot[black!60!white,mark=none,style=thick,dashed] coordinates {(0,0) (5000,0)};
    	\draw[fill=black!60!white,opacity=.2,draw=none] (axis cs:0,-0.1) -- (axis cs:5000,-0.1) -- (axis cs:5000,0.1) -- (axis cs:0,0.1) -- cycle;
	\end{axis}
	\end{tikzpicture}
    \caption{Autocorrelations for the linear regression problem for quantity $\|u\|_{L^2}^2$ for the 4 different sampling algorithms, averaged over particles.}
    \label{fig:linear_ac_alg}
\end{figure}
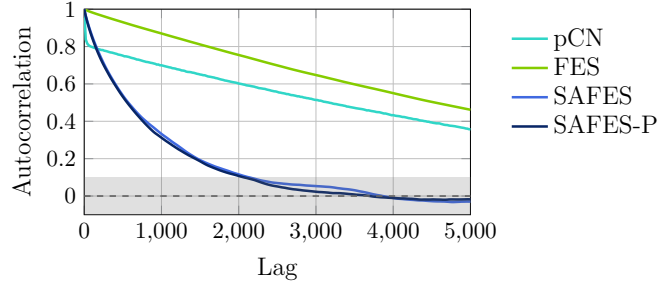

In \cref{fig:linear_densities} we show kernel density estimates for marginals corresponding to two point evaluations, at points $x_1 = 0.01$, $x_2 = 0.63$, for the four different algorithms compared with the true posterior densities, illustrating the accuracy of the methods. We also show (thinned) scatter plots of the point evaluation samples for all particles, illustrating how the different chains mix: the hybrid chains can be seen to be mixing significantly better than the FES chain, which in turn mixes significantly better than the pCN chain. \Cref{tab:linear_errors} compares the sample mean and covariance from the different chains with the true posterior values, along with the multivariate potential scale reduction factor (MPSRF) \cite{brooks1998general}; the latter is computed using inter- and intra-chain correlations, with a value closer to 1 indicating better mixing/convergence. These further illustrate the the mixing properties of the algorithms considered.

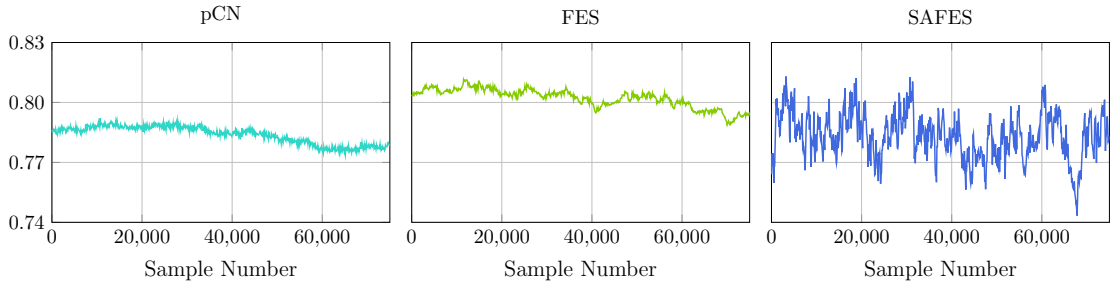
\begin{figure}
    \centering
    \begin{tikzpicture}[scale=0.7]
	\begin{axis}[
   		width=8cm,
    	height=5cm,
    	title={pCN},
    	xlabel={\large Sample Number},
    	xmin=0, xmax=75000,
		ymin=0.74,ymax=0.83,
		ytick={0.74,0.77,0.8,0.83},
		y tick label style={
        /pgf/number format/.cd,
            fixed,
            fixed zerofill,
            precision=2,
        /tikz/.cd
    	},
    	grid=major,
    	scaled x ticks = false, 
		legend style={draw=none,
        	          at={(1.5, 0.95)},
            	      anchor=north},
    	legend style={font=\large},
    	legend cell align=left]
    	\addplot[pcn,mark=none,style=thick]  table [x index=0, y expr=\thisrowno{1}*2*pi, col sep=tab, each nth point={5}] {data/linear_traces.txt};
	\end{axis}
	\end{tikzpicture}
	\begin{tikzpicture}[scale=0.7]
	\begin{axis}[
   		width=8cm,
    	height=5cm,
    	title={FES},
    	xlabel={\large Sample Number},
    	xmin=0, xmax=75000,
    	ymajorticks=false,
		ymin=0.74,ymax=0.83,
		ytick={0.74,0.77,0.8,0.83},
    	grid=major,
		legend style={draw=none,
        	          at={(1.5, 0.95)},
            	      anchor=north},
    	legend style={font=\large},
    	 scaled x ticks = false, 
    	legend cell align=left]
    	\addplot[fes,mark=none,style=thick]  table [x index=0, y expr=\thisrowno{2}*2*pi, col sep=tab, each nth point={5}] {data/linear_traces.txt};
	\end{axis}
	\end{tikzpicture}	
	\begin{tikzpicture}[scale=0.7]
	\begin{axis}[
   		width=8cm,
    	height=5cm,
    	title={\hybrid},
    	xlabel={\large Sample Number},
    	xmin=0, xmax=75000,
    	ymajorticks=false,
		ymin=0.74,ymax=0.83,
		ytick={0.74,0.77,0.8,0.83},
    	grid=major,
    	scaled x ticks = false, 
		legend style={draw=none,
        	          at={(1.5, 0.95)},
            	      anchor=north},
    	legend style={font=\large},
    	legend cell align=left]
    	\addplot[hybrid,mark=none,style=thick]  table [x index=0, y expr=\thisrowno{3}*2*pi, col sep=tab, each nth point={5}] {data/linear_traces.txt};
	\end{axis}
	\end{tikzpicture}
    \caption{Traces of $\|u\|_{L^2}^2$ for the linear regression problem for 3 of the different sampling algorithms.}
    \label{fig:linear_traces}
\end{figure}

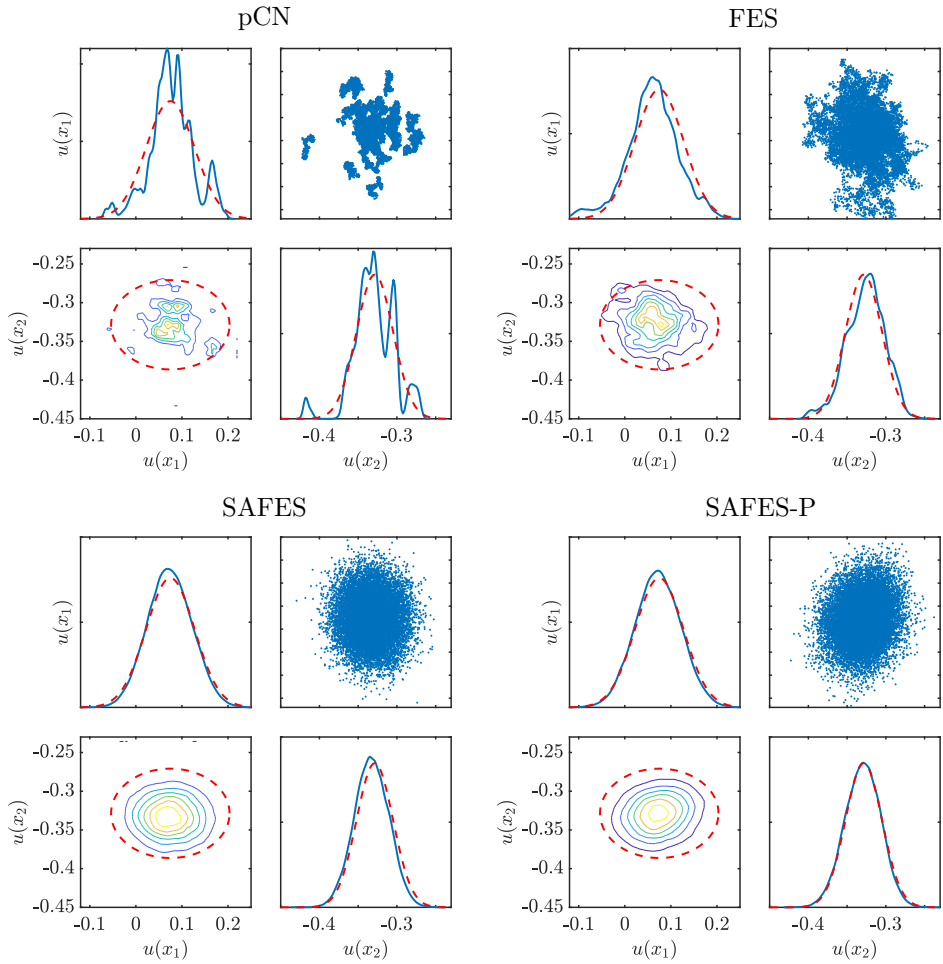
\begin{figure}\centering
    \begin{tikzpicture}
    \node at (-1.5,5) {\includegraphics[width=0.41\textwidth,trim=0cm 0cm 1cm 0cm,clip]{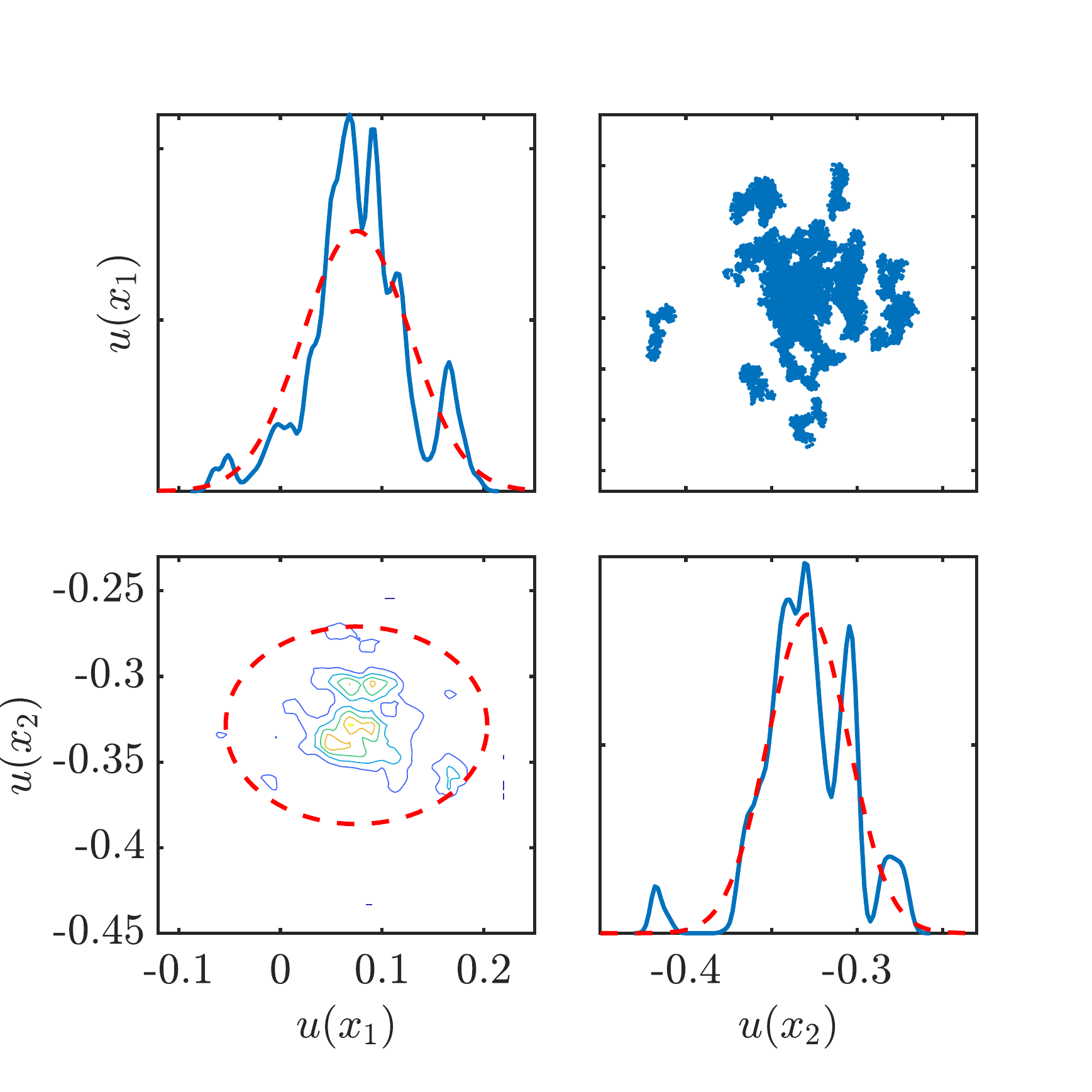}};
    \node at (5,5) {\includegraphics[width=0.41\textwidth,trim=0cm 0cm 1cm 0cm,clip]{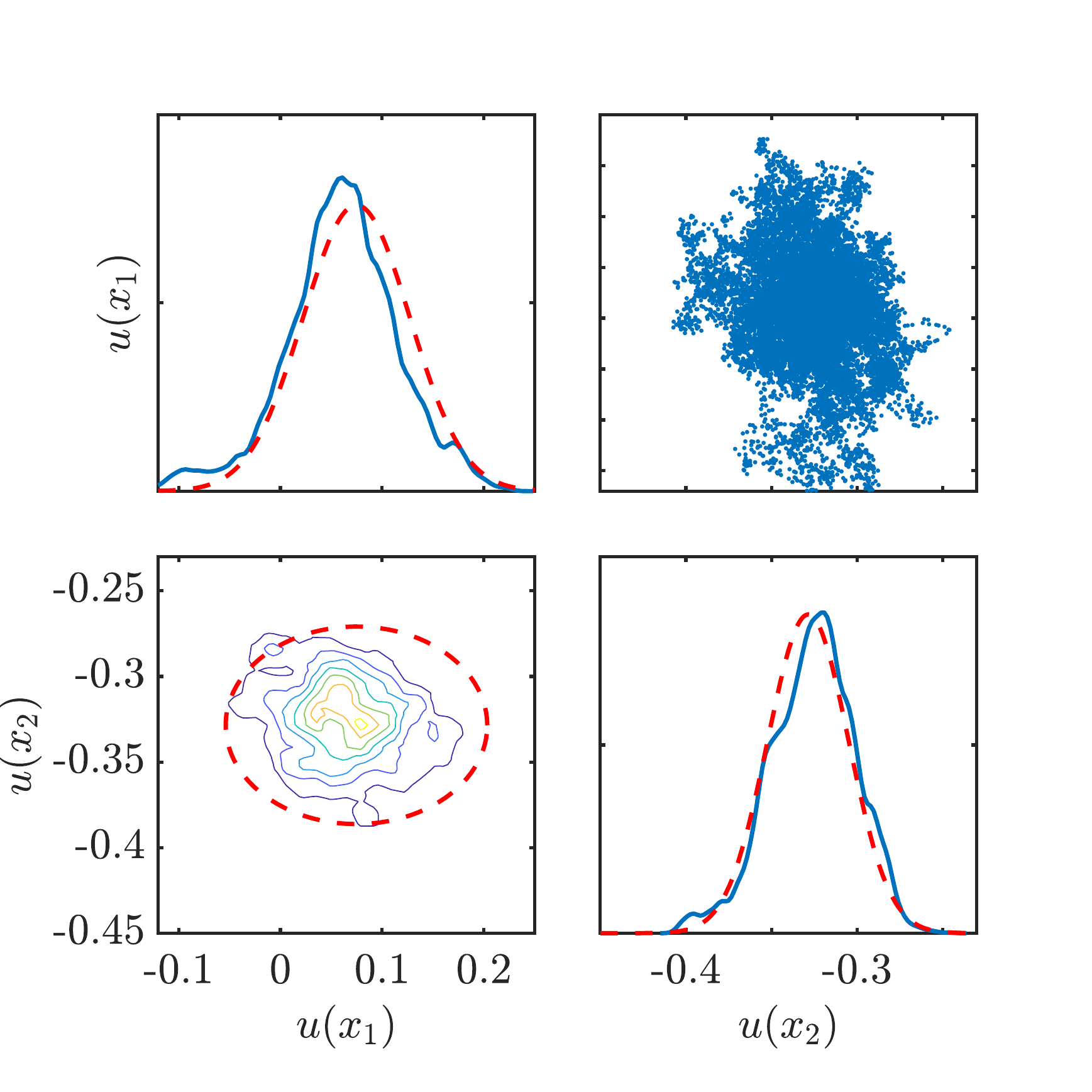}};
    \node at (-1.5,-1.5) {\includegraphics[width=0.41\textwidth,trim=0cm 0cm 1cm 0cm,clip]{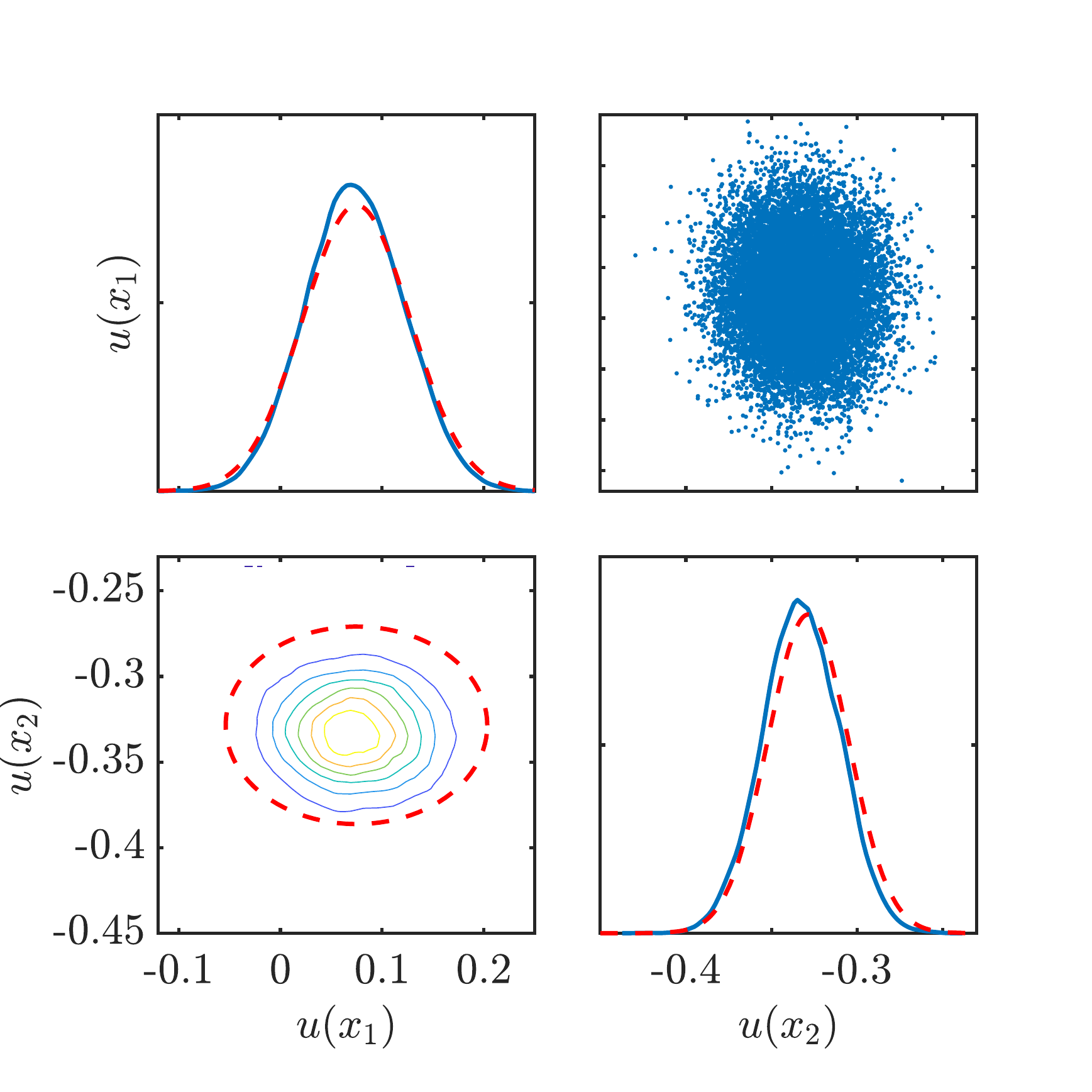}};
    \node at (5,-1.5) {\includegraphics[width=0.41\textwidth,trim=0cm 0cm 1cm 0cm,clip]{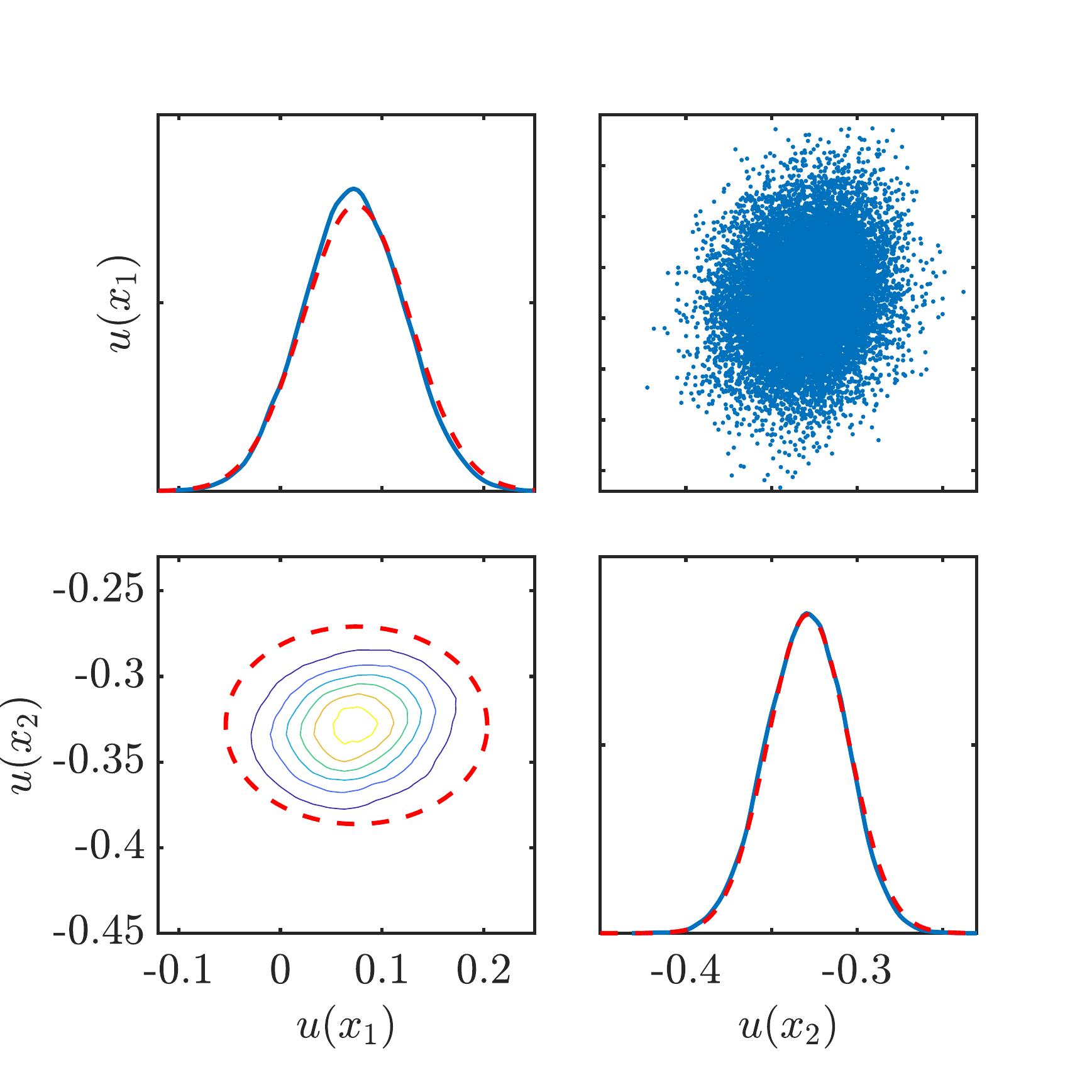}};
    \node at (-1.2,8) {\small pCN};
    \node at (5.3,8) {\small FES};
    \node at (-1.2,1.5) {\small \hybrid};
    \node at (5.4,1.5) {\small \hybridp};
    \end{tikzpicture}
    \caption{Density estimates and scatter plots for the linear regression problem for marginals corresponding to 2 point evaluations, for the 4 different sampling algorithms. The dotted curves represent the true posterior density on the diagonal, and 95\% credible regions below the diagonal.}
    \label{fig:linear_densities}
\end{figure}

\begin{table}
    \centering
    \caption{The relative $\ell^2$ errors in the sample mean $\mathbb{E}(u)$ and sample covariance $\mathrm{Cov}(u)$ for the linear problem, and the MPSRF the 4 different sampling algorithms.}
    \label{tab:linear_errors}
    \begin{tabular}{r|l|l|l}
        & Mean error & Covariance Error & MPSRF\\
        \hline
        pCN & 0.00834 & 0.964 & 17.4\\
        FES & 0.0207  & 0.759 & 4.24\\      
        \hybrid & 0.00645 & 0.404 & 1.074\\
        \hybridp & 0.00784 & 0.390 & 1.075\\
    \end{tabular}
    
\end{table}

\subsubsection{Dependence on number of particles}
As it has been observed that the number of particles $N$ must exceed the dimension of the problem $D$ in order to sample the full posterior when using an affine-invariant sampling method, we study the behavior of the \hybrid algorithm for various numbers of particles $N$. We fix $D = 100$ and vary $N$ between 5 and 40. The number of samples $S$ is varied so that $S\times N = 4\times 10^6$, i.e., the total number of likelihood evaluation remains the same in all cases. The resulting autocorrelations are shown in \cref{fig:linear_ac_npart}. It can be observed that mixing is improved when additional particles are used, however each successive addition of particles yields less of an improvement: the autocorrelation curves accumulate. This is likely related to the effective dimension of the problem being relatively small, as well as the covariance being more accurately estimated. 

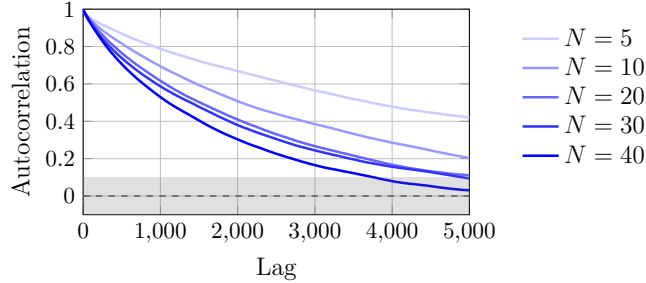
\begin{figure}
    \centering
    \begin{tikzpicture}[scale=0.8]
	\begin{axis}[
   		width=8cm,
    	height=5cm,
    	title={},
    	xlabel={\large Lag},
    	ylabel={\large Autocorrelation},
    	xmin=0, xmax=5000,
    	ymax=1,
    	ymin=-0.1,
    	ytick={0,0.2,0.4,0.6,0.8,1},
    	grid=major,
    	legend style={draw=none,
        	          at={(1.3, 0.95)},
            	      anchor=north},
    	legend style={font=\large},
    	legend cell align=left]
    	\addplot[blue!20,mark=none,style=very thick]  table [x index=0, y index=1, col sep=tab] {data/linear_acs_npart.txt};
    	\addlegendentry{$N=5$}
    	\addplot[blue!40,mark=none,style=very thick]  table [x index=0, y index=2, col sep=tab] {data/linear_acs_npart.txt};
    	\addlegendentry{$N=10$}
    	\addplot[blue!60,mark=none,style=very thick]  table [x index=0, y index=4, col sep=tab] {data/linear_acs_npart.txt};
    	\addlegendentry{$N=20$}
    	\addplot[blue!80,mark=none,style=very thick]  table [x index=0, y index=6, col sep=tab] {data/linear_acs_npart.txt};
    	\addlegendentry{$N=30$}
		\addplot[blue!100,mark=none,style=very thick]  table [x index=0, y index=8, col sep=tab] {data/linear_acs_npart.txt};
    	\addlegendentry{$N=40$}
    	\addplot[black!60!white,mark=none,style=thick,dashed] coordinates {(0,0) (5000,0)};
    	\draw[fill=black!60!white,opacity=.2,draw=none] (axis cs:0,-0.1) -- (axis cs:5000,-0.1) -- (axis cs:5000,0.1) -- (axis cs:0,0.1) -- cycle;
	\end{axis}
	\end{tikzpicture}
    \caption{Autocorrelations of the quantity $\|u\|_{L^2}^2$ for the \hybrid algorithm applied to the linear regression problem, for numbers of particles $N$, averaged over particles.}
    \label{fig:linear_ac_npart}
\end{figure}

\subsubsection{Dependence on dimension}
We now consider the effect of the discretization dimension $D$ on the performance of the \hybrid algorithm. We fix $N = 40$, $K = 10^5$ and vary $D=50\times 2^j$, $j=0,\ldots,4$. The resulting autocorrelations are shown in \cref{fig:linear_ac_dim}. The areas under the curves do not increase with discretization level, suggesting that the statistical performance of the algorithm is dimension-robust. This is in contrast to, for example, the AIES algorithm using the stretch move, which fails to be dimension-robust even when sufficient particles are used to ensure the correct distribution is targeted \cite{huijser2015properties}.

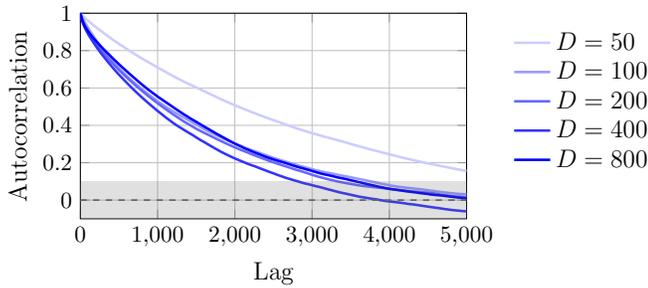
\begin{figure}
    \centering
    \begin{tikzpicture}[scale=0.8]
	\begin{axis}[
   		width=8cm,
    	height=5cm,
    	title={},
    	xlabel={\large Lag},
    	ylabel={\large Autocorrelation},
    	xmin=0, xmax=5000,
    	ymax=1,
    	ymin=-0.1,
    	ytick={0,0.2,0.4,0.6,0.8,1},
    	grid=major,
    	legend style={draw=none,
        	          at={(1.3, 0.95)},
            	      anchor=north},
    	legend style={font=\large},
    	legend cell align=left]
    	\addplot[blue!20,mark=none,style=very thick]  table [x index=0, y index=1, col sep=tab] {data/linear_acs_dim.txt};
    	\addlegendentry{$D=50$}
    	\addplot[blue!40,mark=none,style=very thick]  table [x index=0, y index=2, col sep=tab] {data/linear_acs_dim.txt};
    	\addlegendentry{$D=100$}
    	\addplot[blue!60,mark=none,style=very thick]  table [x index=0, y index=3, col sep=tab] {data/linear_acs_dim.txt};
    	\addlegendentry{$D=200$}
    	\addplot[blue!80,mark=none,style=very thick]  table [x index=0, y index=4, col sep=tab] {data/linear_acs_dim.txt};
    	\addlegendentry{$D=400$}
		\addplot[blue!100,mark=none,style=very thick]  table [x index=0, y index=5, col sep=tab] {data/linear_acs_dim.txt};
    	\addlegendentry{$D=800$}
    	\addplot[black!60!white,mark=none,style=thick,dashed] coordinates {(0,0) (5000,0)};
    	\draw[fill=black!60!white,opacity=.2,draw=none] (axis cs:0,-0.1) -- (axis cs:5000,-0.1) -- (axis cs:5000,0.1) -- (axis cs:0,0.1) -- cycle;
	\end{axis}
	\end{tikzpicture}
    \caption{Autocorrelations of the quantity $\|u\|_{L^2}^2$ for the \hybrid algorithm applied to the linear problem, for different discretization levels $D$, averaged over particles.}
    \label{fig:linear_ac_dim}
\end{figure}

\subsection{A nonlinear inverse problem: Darcy flow}

We now consider a case when the forward map is nonlinear.  We consider an example from \cite{Garbuno-InigoNuskenReich20} for reference, and then consider a modified version with smaller observational noise and a slower decaying prior so that the effective dimension of the problem is increased. Specifically, we consider a one-dimensional Darcy flow problem on spatial domain $\Omega = (0,2\pi)$, defining $\mathcal{S}:L^\infty(\Omega)\to C^0(\Omega)$ as the mapping from $u$ to $p$, where
\[
-\frac{\dee}{\dee x}\left(e^{u(x)}\frac{\dee p}{\dee x}(x)\right) = f(x),\quad x \in \Omega
\]
subject to periodic boundary conditions and $\int_\Omega p(x)\,\dee x = 0$. We make the choice 
\[
f(x) = \exp\left(-\frac{(x-\pi)^2}{10}\right) - c_f,\quad c_f = \int_\Omega \exp\left(-\frac{(x-\pi)^2}{10}\right)\,\dee x
\]
and define the observation operator $\mathcal{O}:C^0(\Omega)\to\R^J$ as in \cref{ssec:linear_numerics}. The nonlinear forward map is then defined by $\G = \mathcal{O}\circ\mathcal{S}$. We assume we have data $y \in \R^J$ arising from the model
\[
y = \G(u) + \eta,\quad \eta \sim N(0,\gamma^2 I)
\]
for some $\gamma > 0$. As in \cite{Garbuno-InigoNuskenReich20}, we take the true state $u^\dagger$ to be $u^\dagger(x) = \sin(x)/2$. A centered Gaussian prior $N(0,C_0)$ is used, and we consider two problems based on this setup:
\begin{enumerate}[(i)]
\item $C_0^{-1} = 4\left(\mu T - \Delta\right)^2$ where $T:L^2 (\Omega)\to L^2(\Omega)$ is given by $Tv = \frac{1}{|\Omega|}\int_\Omega v$, $\Delta$ is the Laplacian with periodic boundary conditions and $\mu = 100$. Moreover, $\gamma = 10^{-2}$ giving a relative error on the data of $3.81\%$.
\item $C_0^{-1} = I-\Delta$ and $\Delta$ is the Laplacian with homogenous Neumann boundary conditions. Moreover, $\gamma = 10^{-4}$ giving a relative error on the data of $0.0381\%$
\end{enumerate}
We observe the solution at $J = 10$ points, $d_j=2\pi j/J$, so that the first problem (i) is identical to the example considered in \cite{Garbuno-InigoNuskenReich20}. Problem (ii) is a modification with a more concentrated posterior, which is more difficult to sample with methods that heavily rely on the prior, such as pCN.
Note that \cite{Garbuno-InigoNuskenReich20} also proposed a gradient-free method for affine-invariant sampling via simulation of a Langevin-type equation; however, this method suffers the same issue as other affine-invariant methods in that the number of particles must exceed the dimension of the problem.

\subsubsection{Comparison of algorithms}
As for the linear case, we compare the four different algorithms on these problems. We fix dimension $D = 100$ and $N=40$ particles, and generate $10^5$ samples per particles, discarding the first $25\%$ as burn-in. Again, for the pCN method independent chains are run for each particle. We first consider problem (i): the autocorrelations for the quantity $\|u\|_{L^2}^2$ are shown in \cref{fig:nonlinear_ac}, kernel density estimates and scatter plots for marginals corresponding to point evaluations at $x_1 = 0.01$, $x_2 = 0.63$ are shown in the top row in \cref{fig:nonlinear_densities}, and MSPRFs are shown in \cref{tab:psrf}. Note that now there is no analytic form for the true posterior densities to compare with as in the linear case. We first note that, since the likelihood is relatively flat, the posterior is not too far from the prior and so a large step size may be used with pCN leading to fast autocorrelation decay. The \hybrid and \hybridp algorithms achieve similar autocorrelation decay, however that for FES is slower.  This is potentially due to FES only using knowledge of the prior eigenmodes but not the decay of its eigenvalues in the space where the AEIS is used, and the prior dominates in this problem. By decreasing the number of modes $M$, better performance could likely be achieved, noting that FES reduces to pCN in the case $M=0$. Nonetheless, the MPSRFs for all algorithms are all close to 1, and the density estimates are similar to one another, since a large number of samples are taken relative to the autocorrelation time.

\begin{table}
    \centering
    \caption{The MPSRFs for the different nonlinear problems for the 4 different sampling algorithms.}
    \label{tab:psrf}
    \begin{tabular}{r|l|l|l}
        Algorithm & Nonlinear (i) & Nonlinear (ii) & Level Set\\
        \hline
        pCN & 1.004 & 12.1 & 34.4\\
        FES & 1.002 & 1.13 & 68.5\\      
        \hybrid & 1.005 & 1.03 & 1.50\\
        \hybridp & 1.008& 1.03& 1.40\\
    \end{tabular}
    
\end{table}

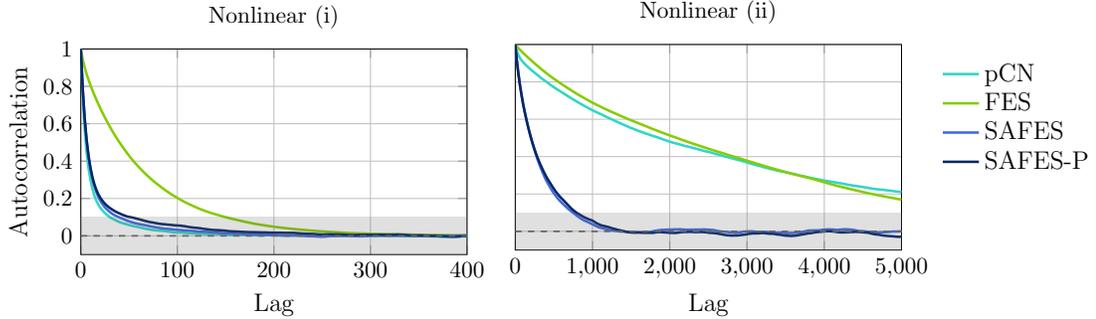
\begin{figure}
    \centering
    \begin{tikzpicture}[scale=0.8]
	\begin{axis}[
   		width=8cm,
    	height=5cm,
    	title={Nonlinear (i)},
    	xlabel={\large Lag},
    	ylabel={\large Autocorrelation},
    	xmin=0, xmax=400,
    	ymax=1,
    	ymin=-0.1,
    	ytick={0,0.2,0.4,0.6,0.8,1},
    	grid=major,
    	legend style={draw=none,
        	          at={(1.5, 0.95)},
            	      anchor=north},
    	legend style={font=\large},
    	legend cell align=left]
    	\addplot[pcn,mark=none,style=very thick]  table [x index=0, y index=1, col sep=tab] {data/nonlinear1_acs_alg.txt};
    	\addplot[fes,mark=none,style=very thick]  table [x index=0, y index=2, col sep=tab] {data/nonlinear1_acs_alg.txt};
    	\addplot[hybrid,mark=none,style=very thick]  table [x index=0, y index=3, col sep=tab] {data/nonlinear1_acs_alg.txt};
    	\addplot[hybrid-proj,mark=none,style=very thick]  table [x index=0, y index=4, col sep=tab] {data/nonlinear1_acs_alg.txt};
    	\addplot[black!60!white,mark=none,style=thick,dashed] coordinates {(0,0) (400,0)};
    	\draw[fill=black!60!white,opacity=.2,draw=none] (axis cs:0,-0.1) -- (axis cs:400,-0.1) -- (axis cs:400,0.1) -- (axis cs:0,0.1) -- cycle;
	\end{axis}
	\end{tikzpicture}
	\begin{tikzpicture}[scale=0.8]
	\begin{axis}[
   		width=8cm,
    	height=5cm,
    	title={Nonlinear (ii)},
    	xlabel={\large Lag},
    	xmin=0, xmax=5000,
    	ymax=1,
    	ymin=-0.1,
    	ytick={0,0.2,0.4,0.6,0.8,1},
		ymajorticks=false,
    	grid=major,
    	legend style={draw=none,
        	          at={(1.3, 0.95)},
            	      anchor=north},
    	legend style={font=\large},
    	legend cell align=left]
    	\addplot[pcn,mark=none,style=very thick]  table [x index=0, y index=1, col sep=tab] {data/nonlinear2_acs_alg.txt};
    	\addlegendentry{pCN}
    	\addplot[fes,mark=none,style=very thick]  table [x index=0, y index=2, col sep=tab] {data/nonlinear2_acs_alg.txt};
    	\addlegendentry{FES}
    	\addplot[hybrid,mark=none,style=very thick]  table [x index=0, y index=3, col sep=tab] {data/nonlinear2_acs_alg.txt};
    	\addlegendentry{\hybrid}
    	\addplot[hybrid-proj,mark=none,style=very thick]  table [x index=0, y index=4, col sep=tab] {data/nonlinear2_acs_alg.txt};
    	\addlegendentry{\hybridp}
    	\addplot[black!60!white,mark=none,style=thick,dashed] coordinates {(0,0) (5000,0)};
    	\draw[fill=black!60!white,opacity=.2,draw=none] (axis cs:0,-0.1) -- (axis cs:5000,-0.1) -- (axis cs:5000,0.1) -- (axis cs:0,0.1) -- cycle;
	\end{axis}
	\end{tikzpicture}
	
    \caption{Autocorrelations of the quantity $\|u\|_{L^2}^2$ for the nonlinear problems (i), (ii) for the 4 different sampling algorithms, averaged over particles.}
    \label{fig:nonlinear_ac}
\end{figure}

For problem (ii) the corresponding autocorrelations and density estimates are shown in \cref{fig:nonlinear_ac} and the bottom row in \cref{fig:nonlinear_densities}. The autocorrelation behavior is similar to the linear case, with \hybrid and \hybridp performing similarly to each other and outperforming both pCN and FES. Again the pCN autocorrelation decays faster than FES initially, but FES is faster asymptotically.  The density estimates and scatter plots illustrate the poor mixing of pCN compared to the other algorithms. Note that even though the autocorrelation for FES decays only slightly faster than for pCN, the mixing appears much better than pCN and close to that for the hybrid algorithms. The MPSRFs in \cref{tab:psrf} mirror this, with all algorithms significantly outperforming pCN and the hybrid algorithms outperforming FES.

\begin{figure}[bt]\centering
    \begin{tikzpicture}
    \node at (-5,5) {\includegraphics[width=0.32\textwidth,trim=0cm 0.5cm 1.5cm 0cm,clip]{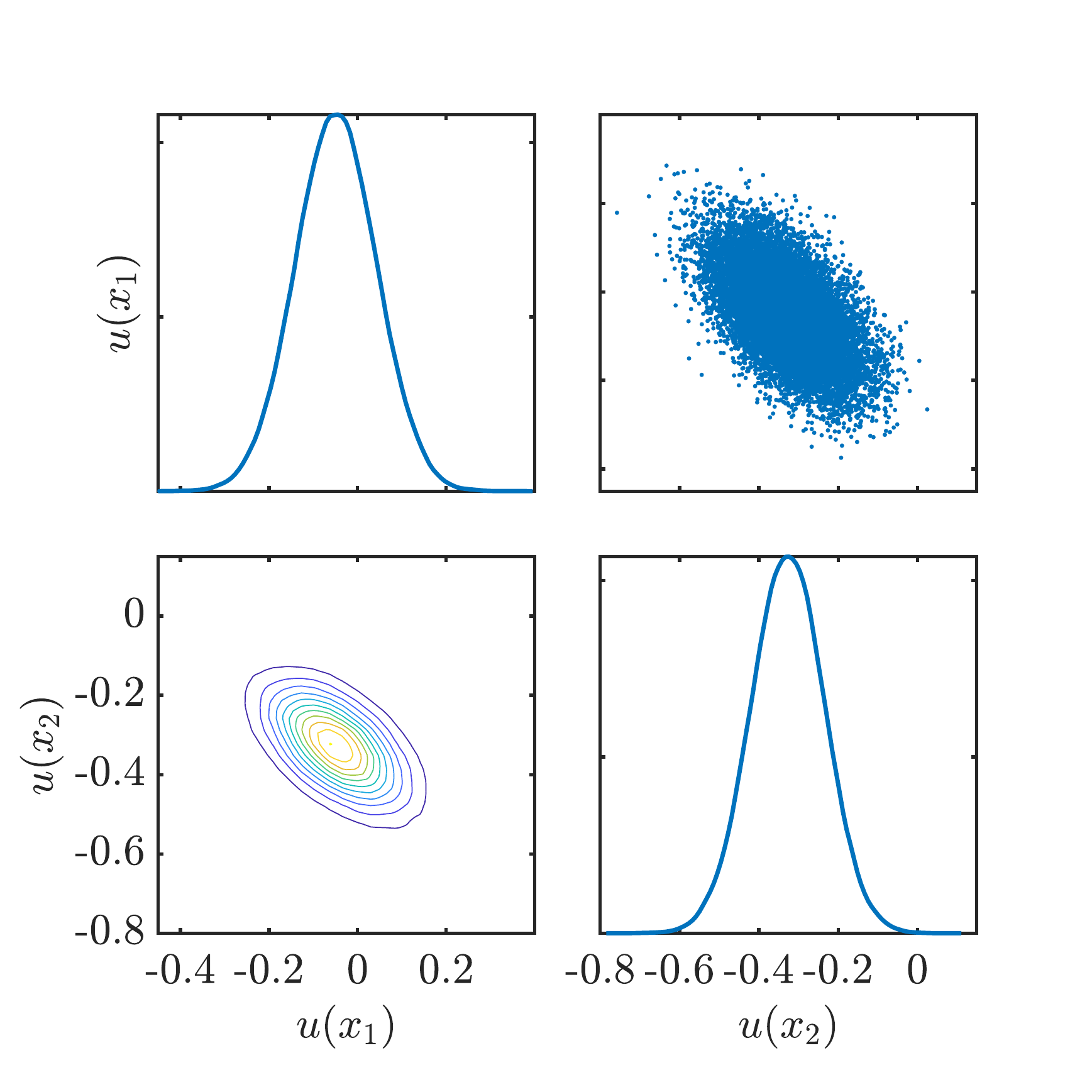}};
    \node at (0,5) {\includegraphics[width=0.32\textwidth,trim=0cm 0.5cm 1.5cm 0cm,clip]{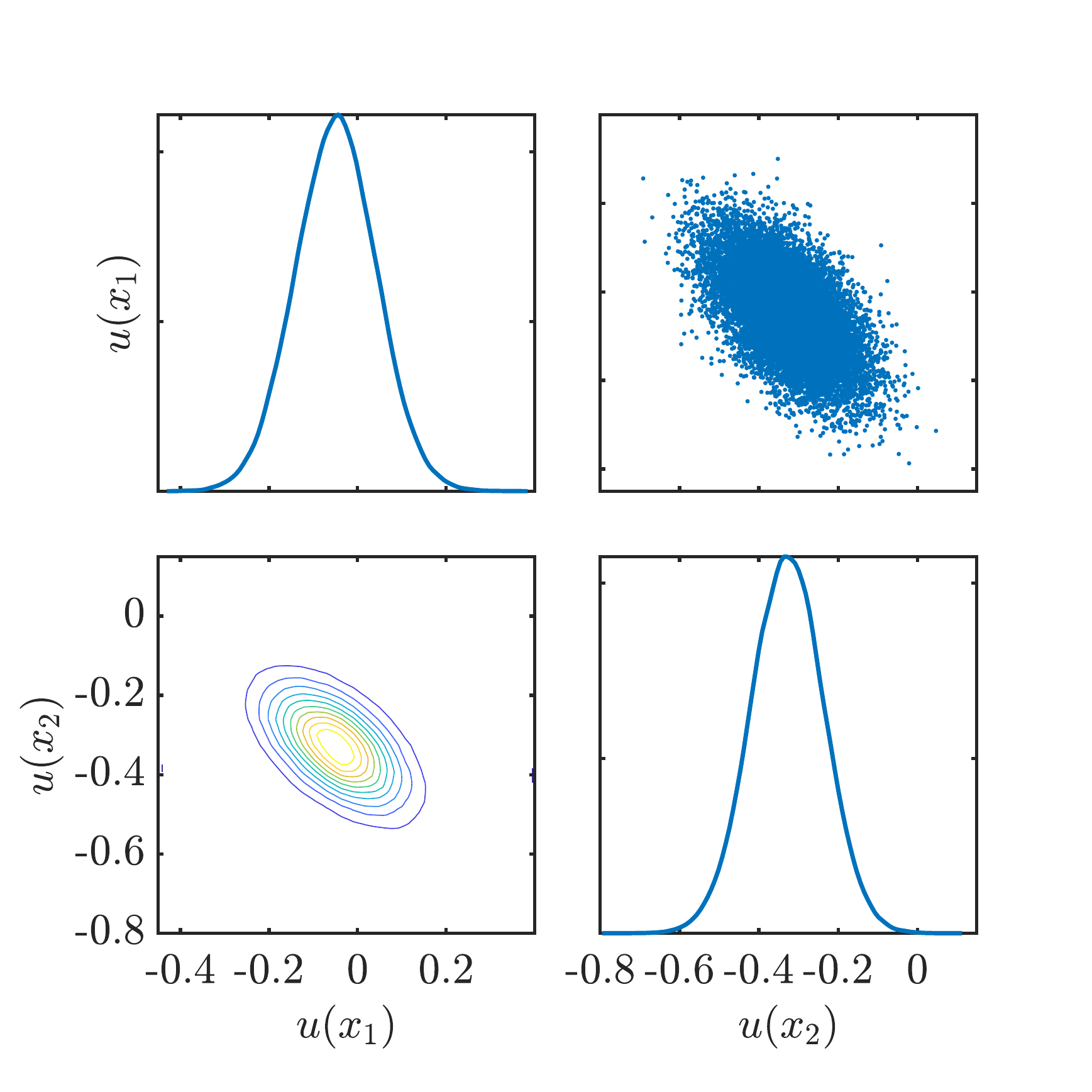}};
    \node at (5,5) {\includegraphics[width=0.32\textwidth,trim=0cm 0.5cm 1.5cm 0cm,clip]{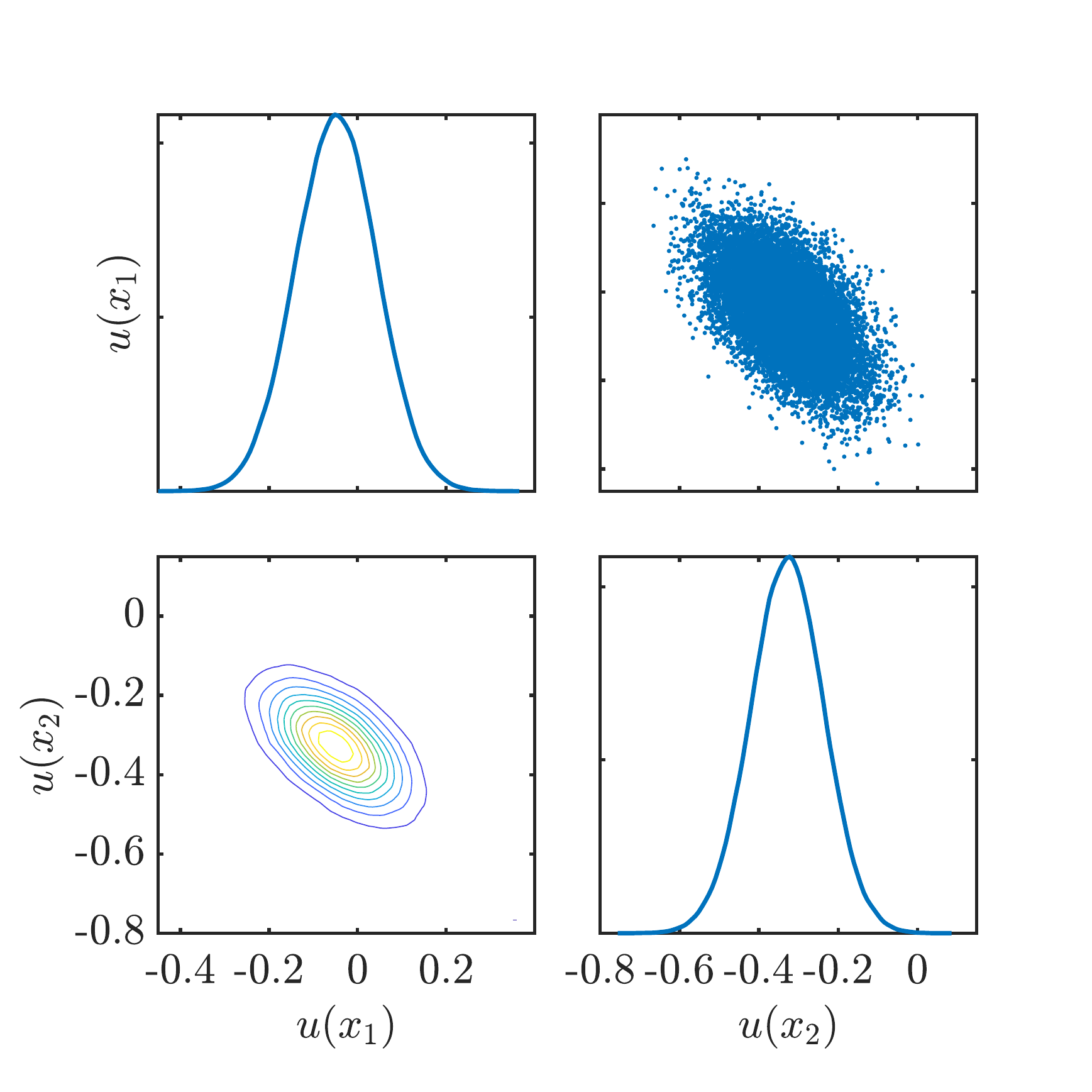}};
    \node at (-4.7,7.3) {\small pCN for (i)};
    \node at (0.3,7.3) {\small FES for (i)};
    \node at (5.3,7.3) {\small \hybrid for (i)};
    \end{tikzpicture}\\

    \begin{tikzpicture}
    \node at (-5,5) {\includegraphics[width=0.32\textwidth,trim=0cm 0.5cm 1.5cm 0cm,clip]{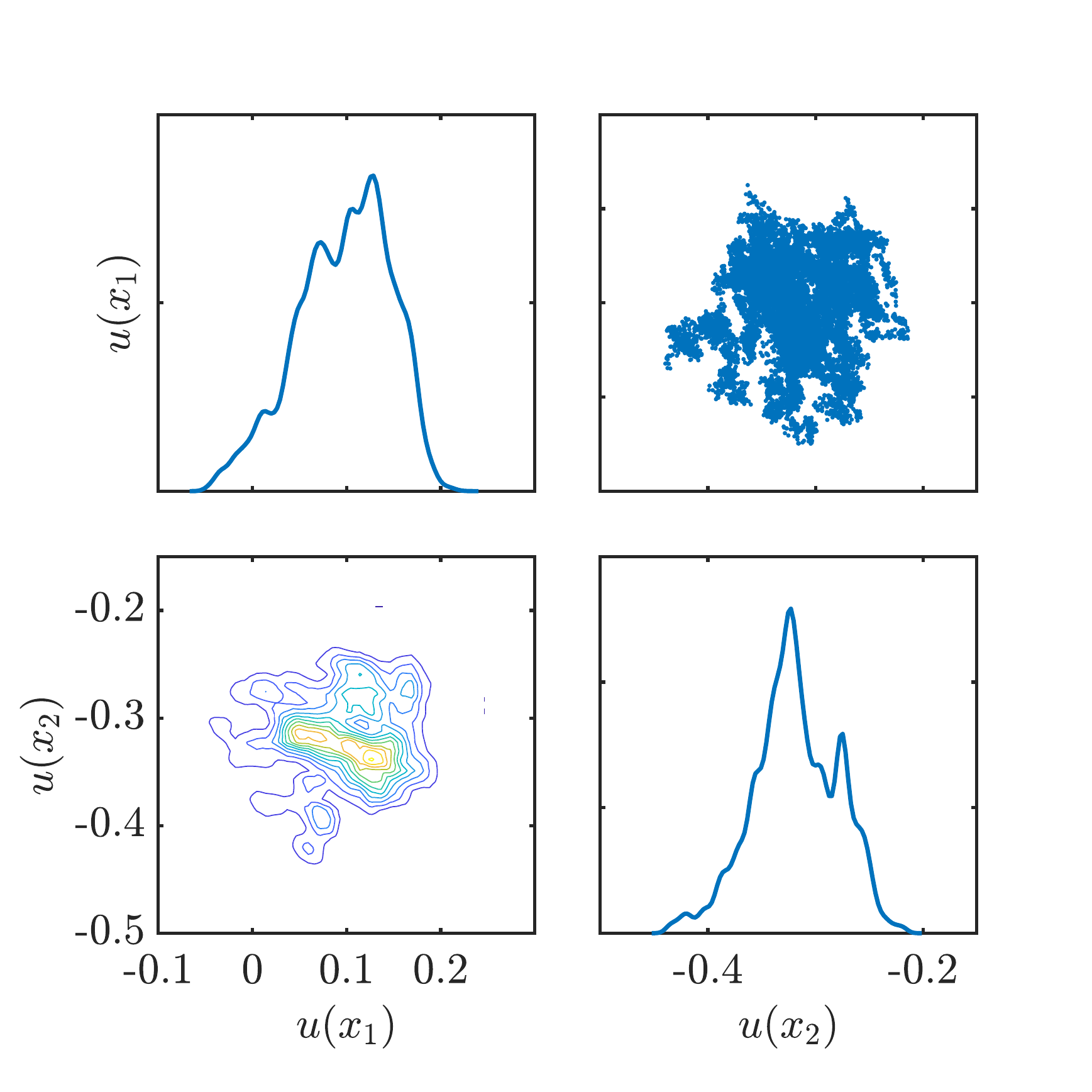}};
    \node at (0,5) {\includegraphics[width=0.32\textwidth,trim=0cm 0.5cm 1.5cm 0cm,clip]{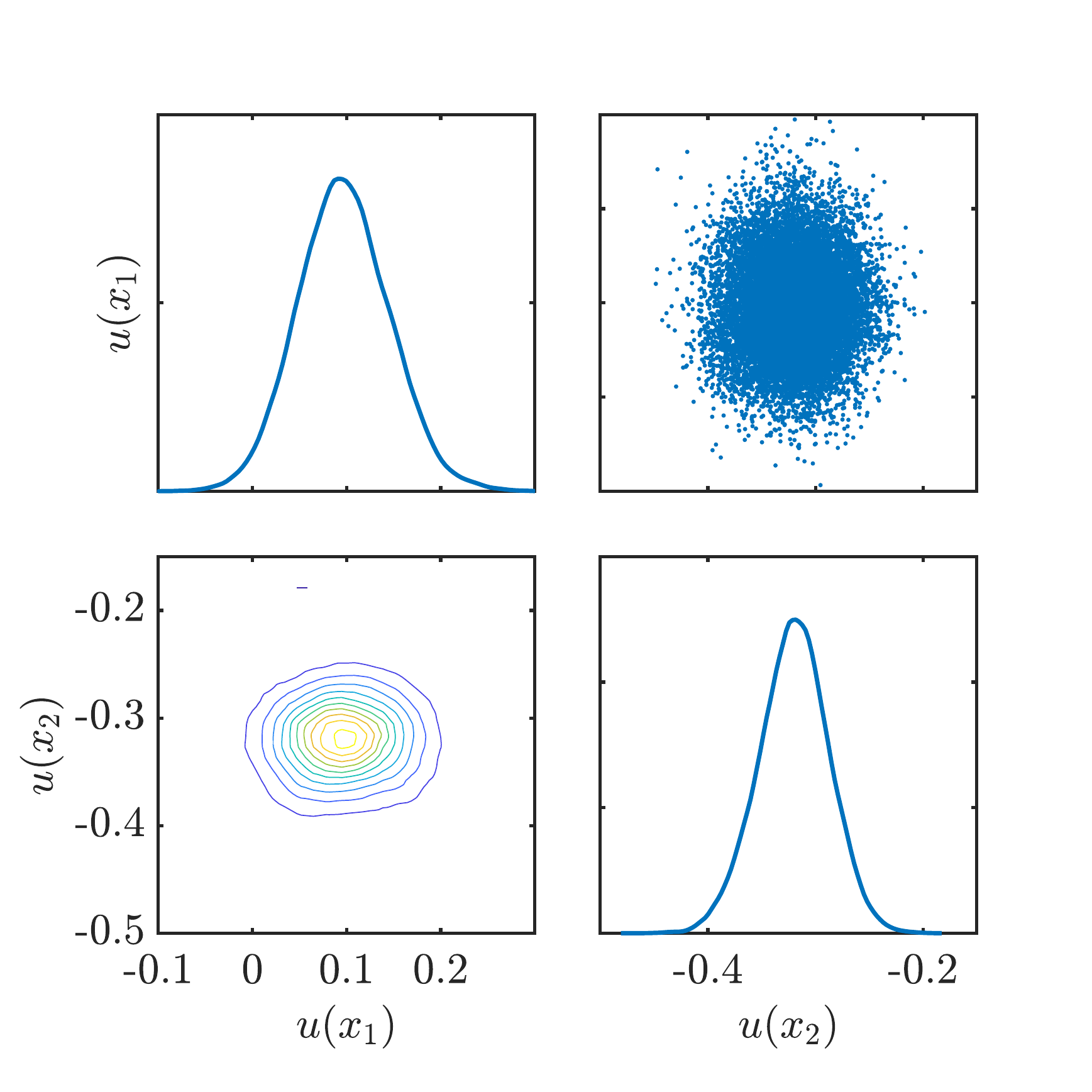}};
    \node at (5,5) {\includegraphics[width=0.32\textwidth,trim=0cm 0.5cm 1.5cm 0cm,clip]{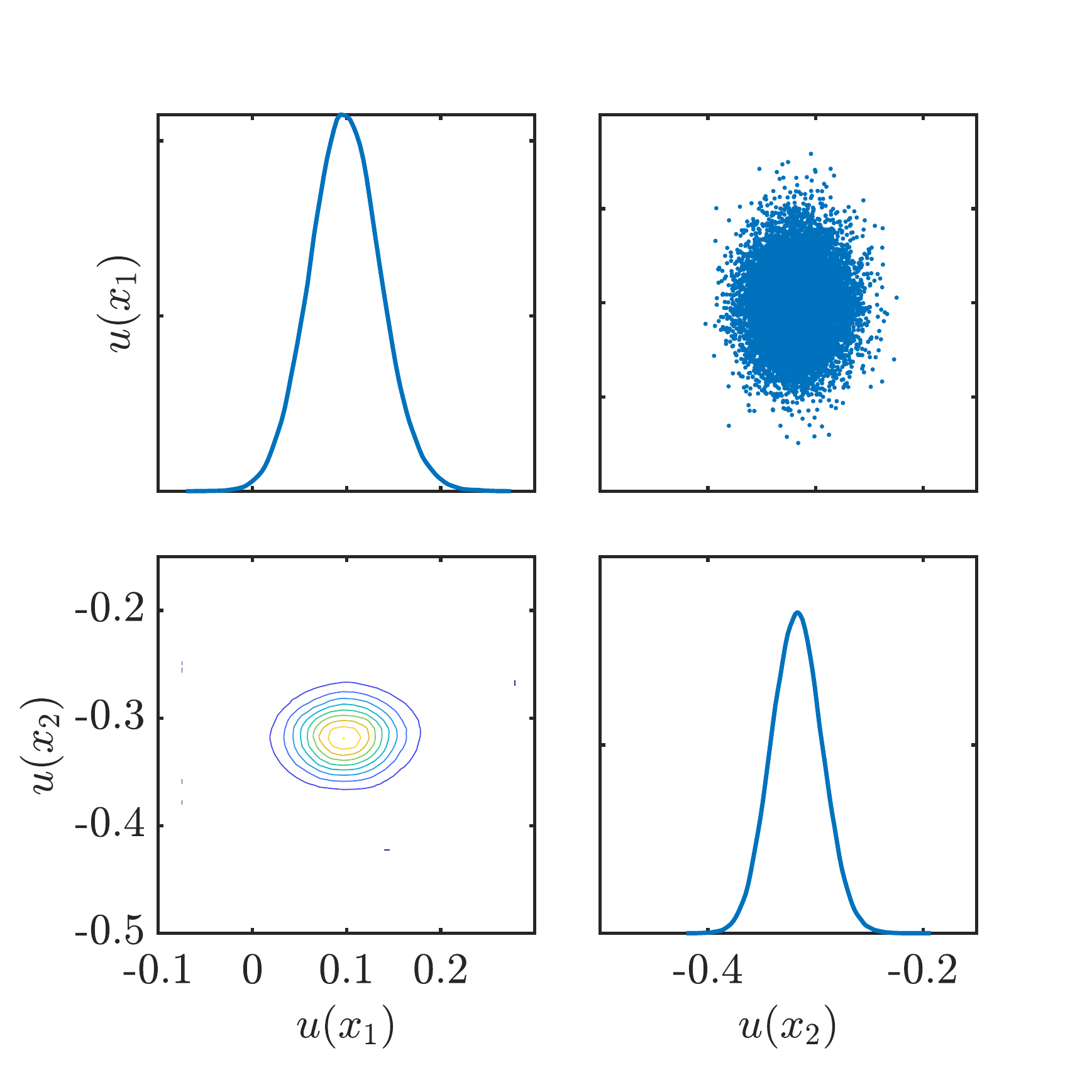}};
    \node at (-4.7,7.3) {\small pCN for (ii)};
    \node at (0.3,7.3) {\small FES for (ii)};
    \node at (5.3,7.3) {\small \hybrid for (ii)};
    \end{tikzpicture}%
    \caption{Density estimates and scatter plots for marginals corresponding to 2 different point evaluations for the nonlinear problem, for 3 different sampling algorithms. The upper row corresponds to setup (i), and the lower row to setup (ii).}
    \label{fig:nonlinear_densities}
\end{figure}
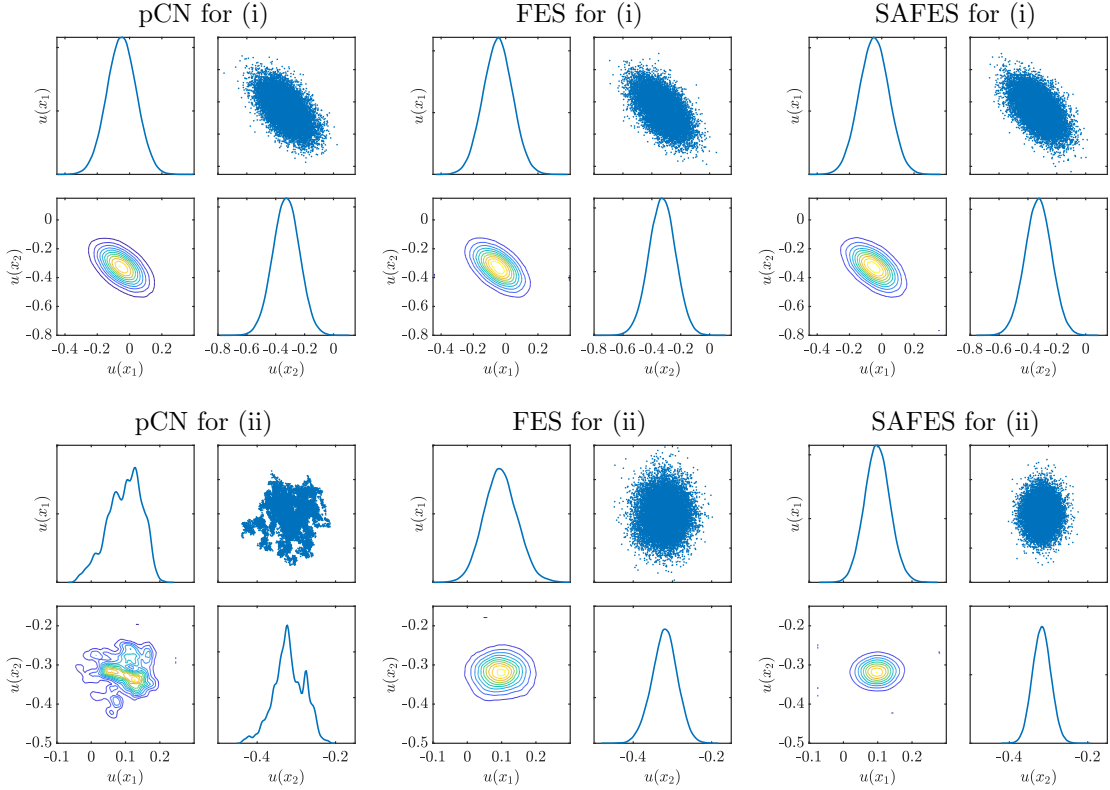

\subsection{A non-differentiable problem: level set prior} We finally consider an example where derivatives do not exist and so gradient-based methods are unavailable. Specifically, we consider a linear inverse problem with a level set prior \cite{iglesias2016bayesian,dunlop2017hierarchical}, with the intention of recovering a piecewise constant field. Whilst the forward map is linear, the level set mapping included in the likelihood ensures that the posterior distribution is non-Gaussian. Specifically, let $\Omega = (0,1)^2$ and define the map $\mathcal{S}:L^2(\Omega)\to C^0(\Omega)$, $u\mapsto p$,  
\[
-\Delta p(x) = \mathrm{sgn}(u(x)),\quad x \in \Omega
\]
subject to homogeneous Dirichlet boundary conditions. Define the observation operator $\mathcal{O}:C^0(\Omega)\to \R^J$ as point evaluations on a uniform grid of $J=9$ points, and the nonlinear forward map $\mathcal{G} = \mathcal{O}\circ \mathcal{S}$. The data $y \in \R^J$ is assumed to arise from the model
\[
y = \mathcal{G}(u) + \eta,\quad \eta \sim N(0,\gamma^2 I)
\]
with $\gamma = 10^{-3}$. A continuous Gaussian prior $\mu_0 = N(0,C_0)$ is placed on $u$, $C_0 = (I-\Delta)^2$, with the intention of recovering the binary field $\mathrm{sgn}(u)$. The true binary field is the indicator function of a circle, with the domain discretized on a uniform mesh of $D = 32^2$ points. We fix $N = 80$ particles; for the FES method we choose the number of modes $M = 10$ and for the \hybridp method we choose $M = 20$. For both \hybrid and \hybridp methods we fix $\lambda = 0.2$ as previously. \Cref{fig:levelset_ac} shows the resulting autocorrelations for the quantity $\|u\|_{L^2}^2$, \cref{fig:levelset_densities} shows kernel density estimates for marginals corresponding to two point observations of the field $u$, at points $x_1 = (0.03,0.03)$, $x_2 = (0.5,0.56)$, and \cref{tab:psrf} shows the resulting MSPRFs. The same trends as for the nonlinear problem (ii) are observed, though in this case the FES mixes significantly more slowly.

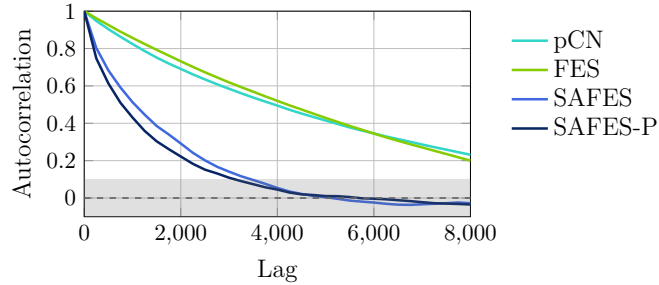
\begin{figure}
    \centering
    \begin{tikzpicture}[scale=0.8]
	\begin{axis}[
   		width=8cm,
    	height=5cm,
    	title={},
    	xlabel={\large Lag},
    	ylabel={\large Autocorrelation},
    	xmin=0, xmax=8000,
    	ymax=1,
    	ymin=-0.1,
    	ytick={0,0.2,0.4,0.6,0.8,1},
    	grid=major,
    	legend style={draw=none,
        	          at={(1.3, 0.95)},
            	      anchor=north},
    	legend style={font=\large},
    	legend cell align=left]
    	\addplot[pcn,mark=none,style=very thick]  table [x index=0, y index=1, col sep=tab] {data/level_acs_alg.txt};
    	\addlegendentry{pCN}
    	\addplot[fes,mark=none,style=very thick]  table [x index=0, y index=2, col sep=tab] {data/level_acs_alg.txt};
    	\addlegendentry{FES}
    	\addplot[hybrid,mark=none,style=very thick]  table [x index=0, y index=3, col sep=tab] {data/level_acs_alg.txt};
    	\addlegendentry{\hybrid}
    	\addplot[hybrid-proj,mark=none,style=very thick]  table [x index=0, y index=4, col sep=tab] {data/level_acs_alg.txt};
    	\addlegendentry{\hybridp}
    	\addplot[black!60!white,mark=none,style=thick,dashed] coordinates {(0,0) (8000,0)};
    	\draw[fill=black!60!white,opacity=.2,draw=none] (axis cs:0,-0.1) -- (axis cs:8000,-0.1) -- (axis cs:8000,0.1) -- (axis cs:0,0.1) -- cycle;
	\end{axis}
	\end{tikzpicture}
    \caption{Autocorrelations for the level set problem for 3 different point evaluations, for the 4 different sampling algorithms, averaged over particles.}
    \label{fig:levelset_ac}
\end{figure}

\begin{figure}\centering
    \begin{tikzpicture}
    \node at (-5,5) {\includegraphics[width=0.32\textwidth,trim=0cm 0.5cm 1.5cm 0cm,clip]{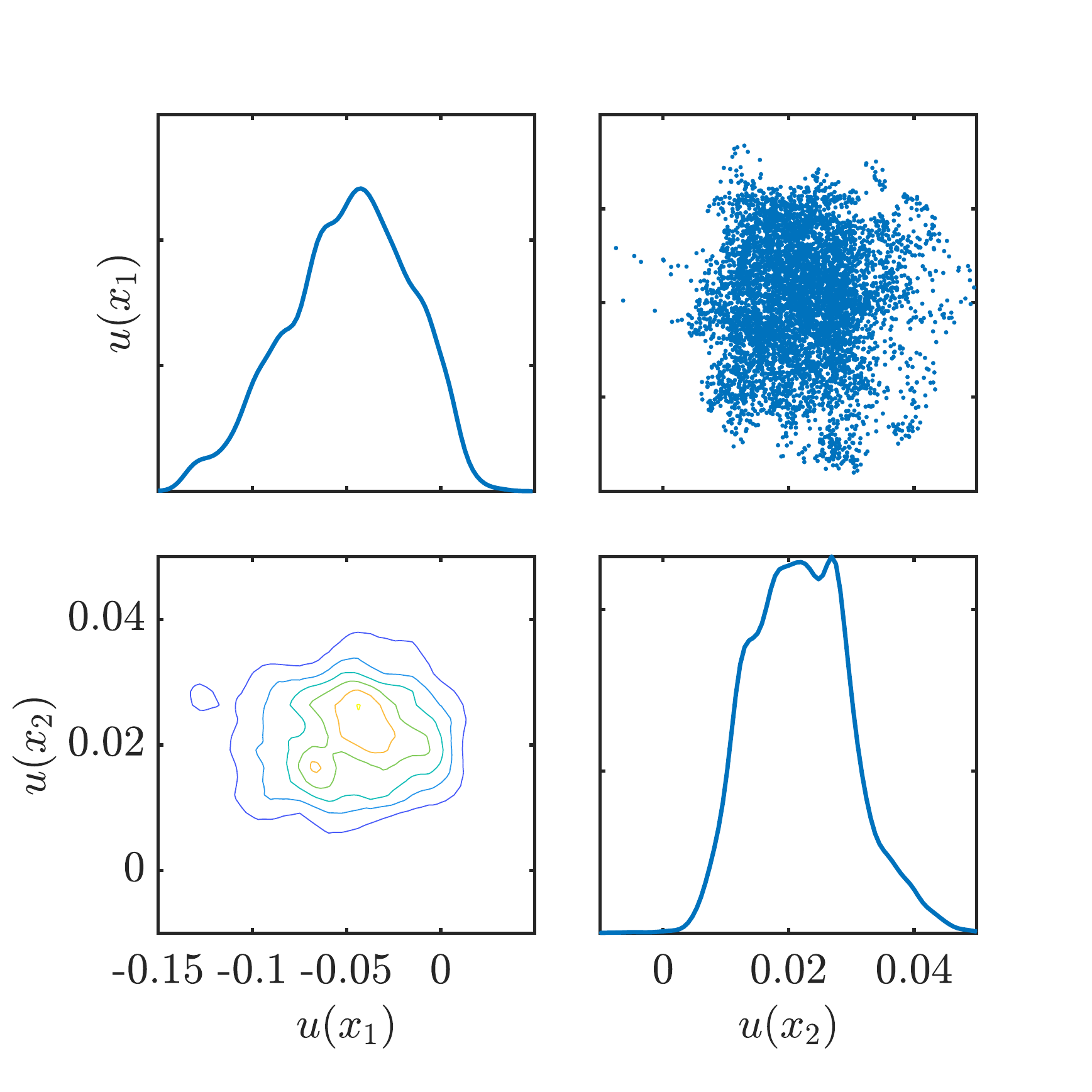}};
    \node at (0,5) {\includegraphics[width=0.32\textwidth,trim=0cm 0.5cm 1.5cm 0cm,clip]{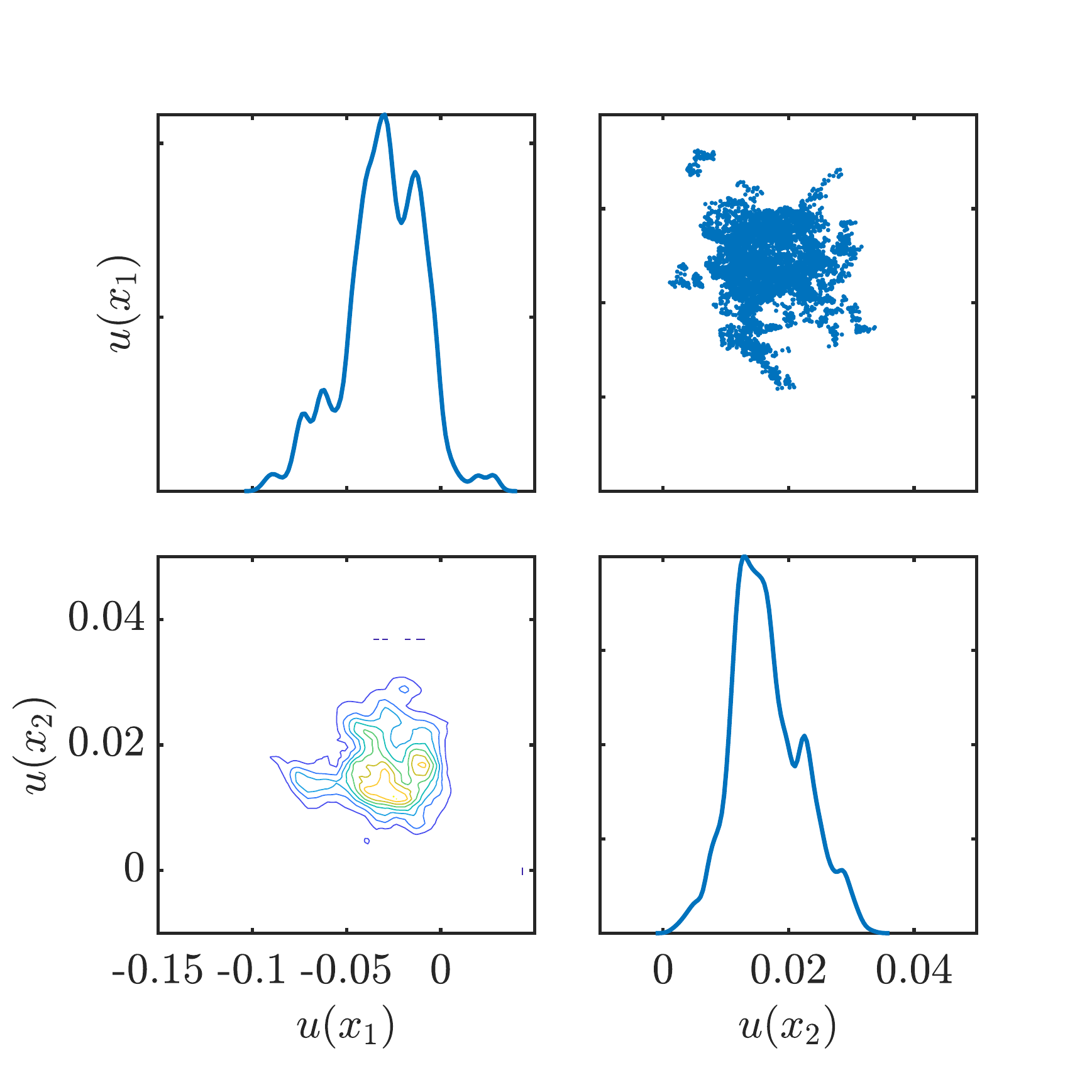}};
    \node at (5,5) {\includegraphics[width=0.32\textwidth,trim=0cm 0.5cm 1.5cm 0cm,clip]{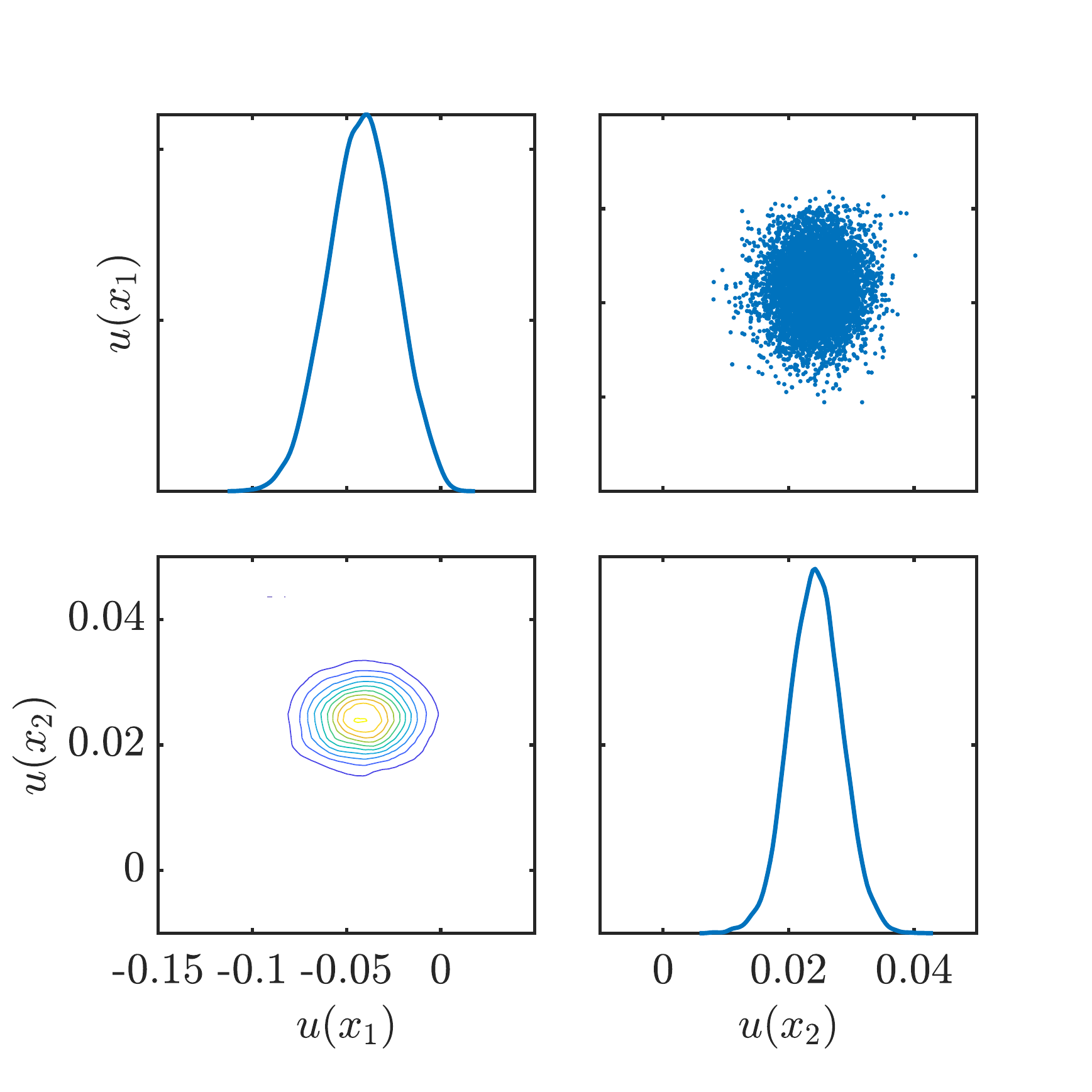}};
    \node at (-4.7,7.3) {\small pCN};
    \node at (0.3,7.3) {\small FES};
    \node at (5.3,7.3) {\small \hybrid};
    \end{tikzpicture}%
    \caption{Density estimates and scatter plots for marginals corresponding to 2 different point evaluations for the level set problem, for 3 different sampling algorithms. }
    \label{fig:levelset_densities}
\end{figure}
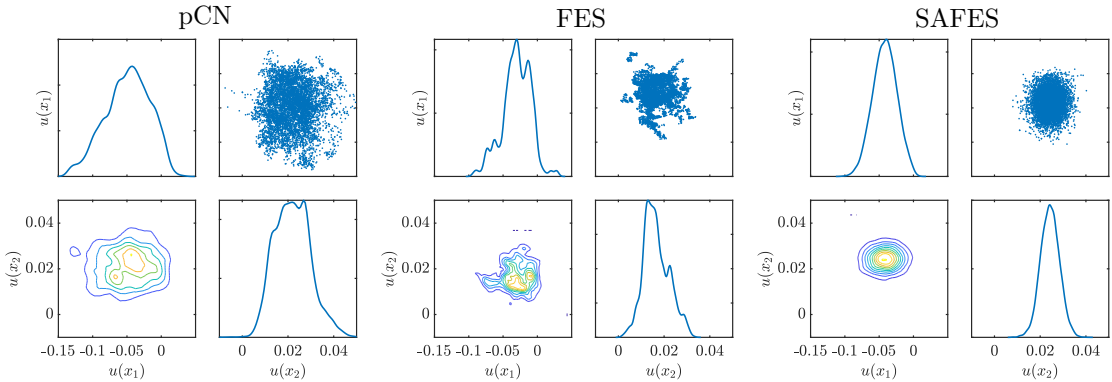

\section{Conclusions}
By combining affine-invariant with dimension-robust sampling methods, one can find a compromise between the advantages and disadvantages of both. Specifically, in the context of Bayesian inverse problems, when the data is particularly informative and the unknown state is high-dimensional, one can obtain a viable method of sampling the posterior distribution without the need for derivatives of the likelihood.

\appendix
\section{Proofs}\label{appendix}

\begin{proposition}
Define $I_C:X\to\R$ by \cref{eq:IC}. The each term in this expression is finite almost-surely under the posterior.
\end{proposition}

\begin{proof}
As the posterior is absolutely continuous with respect to the prior, it suffices to show that the terms are finite almost-surely under any measure equivalent to the prior. The finiteness of the final two terms follows from the Cameron-Martin theorem applied to the measures $N(m,C)$ and $N(0,C)$, as this is simply the logarithm of the Radon-Nikodym derivative between them. For the first term, note that we have for any Hilbert-Schmidt operator $Z:X\to X$ and $u \sim N(0,C_0)$,
\begin{align*}
\mathbb{E}\langle u,Zu\rangle_{C_0}^2 &= \mathbb{E}\langle \xi,Z\xi\rangle_X^2,\quad \xi \sim N(0,I)\\
&= \mathbb{E}\sum_{i,j,k,l} \xi_i\xi_j\xi_k\xi_l \langle \varphi_i,Z\varphi_j\rangle_X \langle \varphi_k,Z\varphi_l\rangle_X,\quad \xi_j \iid N(0,1)\\
&= 3 \sum_j \langle \varphi_j,Z\varphi_j\rangle_X^2\\
&= 3\|Z\|_{HS}^2,
\end{align*}
where $\{\varphi_j\}$ is any orthonormal basis for $X$. The operator $I-C_0^{\frac{1}{2}}C^{-1}C_0^{\frac{1}{2}}$ is Hilbert-Schmidt by the assumed equivalence of $N(0,C_0)$ and $N(0,C)$ and the Feldman-Hajek theorem, so the result follows.
\end{proof}


\end{document}